\documentclass{article}
\usepackage{fullpage}
\usepackage{graphicx} 
\usepackage{amsmath}
\usepackage{amsfonts}
\usepackage{physics}
\usepackage[textsize=tiny]{todonotes}
\usepackage{hyperref}
\usepackage{amsthm}
\usepackage{cleveref}
\usepackage{authblk}
\usepackage{comment}
\usepackage{amssymb}
\usepackage{natbib}

\usepackage{algorithm}

\newcommand{\davg}{\mathcal{D}_\text{avg}}
\newcommand{\gpid}{\mathcal{D}_P}
\newcommand{\frobd}{\mathcal{D}_F}
\newcommand{\supp}{\textnormal{supp}}

\newcommand{\calC}{\mathcal{C}}

\newcommand{\calE}{\mathcal{E}}
\newcommand{\calF}{\mathcal{F}}

\newcommand{\calM}{\mathcal{M}}
\newcommand{\calN}{\mathcal{N}}
\newcommand{\calO}{\mathcal{O}}
\newcommand{\calP}{\mathcal{P}}
\newcommand{\calU}{\mathcal{U}}
\newcommand{\calS}{\mathcal{S}}

\newcommand{\AC}{$\mathsf{AC}$}
\newcommand{\NC}{$\mathsf{NC}$}
\newcommand{\NCZ}{$\mathsf{NC}^0$}
\newcommand{\ACZ}{$\mathsf{AC}^0$}
\newcommand{\QAC}{$\mathsf{QAC}$}
\newcommand{\QNC}{$\mathsf{QNC}$}
\newcommand{\RLWE}{$\mathsf{RLWE}$}
\newcommand{\QACZ}{$\mathsf{QAC}^0$}
\newcommand{\QACZf}{$\mathsf{QAC}^0_f$}
\newcommand{\QNCZ}{$\mathsf{QNC}^0$}
\newcommand{\QNCO}{$\mathsf{QNC}^1$}
\newcommand{\wt}{\textbf{W}}

\usepackage{cleveref}
\usepackage{multirow}
\usepackage{colortbl}
\usepackage{physics}

\usepackage[english]{babel}
\usepackage{amsmath}
\usepackage{amsthm}
\usepackage{amsthm,amsfonts,amssymb,amsmath} 
\newtheorem{theorem}{Theorem}
\newtheorem{conjecture}{Conjecture}
\newtheorem{proposition}{Proposition}
\newtheorem{lemma}{Lemma}

\newtheorem{corollary}{Corollary}
\newtheorem{definition}{Definition}
\newtheorem{fact}{Fact}
\newtheorem{openprob}{Open Problem}

\newcommand{\obsanc}{O_{P_i, n+a}} 
\newcommand{\obsared}{O_{P_i, n}} 

\newcommand{\support}{\ell}
\newcommand{\gaterem}{\kappa}

\newcommand{\ex}{\mathbb{E}}

\newcommand{\paulis}{\mathcal{P}^n}
\newcommand{\pauli}{\mathcal{P}}
\newcommand{\poly}{\textnormal{poly}}
\newcommand{\proj}{\textnormal{Proj}}

\title{Learning shallow quantum circuits with many-qubit gates}
\author[1]{Francisca Vasconcelos\footnote{\url{francisca@berkeley.edu}}}
\author[2, 3]{Hsin-Yuan Huang\footnote{ \url{hsinyuan@caltech.edu}}}
\affil[1]{University of California, Berkeley}
\affil[2]{Google Quantum AI}
\affil[3]{California Institute of Technology}
\date{}

\begin{document}

\maketitle
\begingroup
\renewcommand{\thefootnote}{}%
\footnotetext{\newline \textsf{Accepted for presentation at the Conference on Learning Theory (COLT) 2025.}}
\addtocounter{footnote}{-1}%
\endgroup

\begin{abstract}%
  The seminal work of [LMN93] established a cornerstone result for classical complexity, with profound implications for learning theory. By proving low-degree Fourier concentration of \ACZ, the work demonstrated that Boolean functions computed by constant-depth circuits can be efficiently PAC-learned via low-degree Fourier sampling. This breakthrough provided the first sample- and time-efficient (quasi-polynomial) algorithm for learning \ACZ. 
  
  Proposed by [Moore99] as a natural quantum analog of \ACZ, \QACZ~is the class of constant-depth quantum circuits composed of arbitrary single-qubit gates and polynomial $CZ$ gates of unbounded width. In this work, we present the first algorithm for efficient average-case learning of \QACZ~circuits with logarithmic ancilla. Namely, our algorithm achieves quasi-polynomial sample- and time-complexity for learning unknown \QACZ~unitaries to inverse-polynomially small error. We further show that these learned unitaries can be efficiently synthesized via poly-logarithmic depth circuits, making progress towards proper learning of \QACZ. Since in finite-dimensional circuit geometries \QACZ~circuits require polynomial depth to implement, this result significantly expands the family of efficiently learnable quantum circuits. 
\end{abstract}

\newpage
\tableofcontents

\newpage
\section{Introduction}

Motivated by both theory and practice, shallow quantum circuits have become central to quantum complexity. Theoretically, shallow quantum circuits have proven more powerful than their classical counterparts \citep{bravyi2018quantum,watts2019exponential, bravyi2020quantum, watts2023unconditional}, with output distributions that are expected to be classically hard to simulate \citep{terhal2002adaptive, gao2017quantum, bermejo2018architectures, haferkamp2020closing, hangleiter2023computational}.
Experimentally, current quantum hardware is noisy and only capable of simulating short time-scale quantum processes. In order to characterize and continue improving quantum hardware, it is crucial that we develop efficient learning procedures for these shallow quantum circuits.

Classically, efficient learning of shallow circuits was an immediate consequence of a seminal result in classical circuit complexity. Namely, \cite{linial1993constant} established low-degree Fourier concentration of Boolean functions implementable in \ACZ---the family of constant-depth, polynomial-sized circuits with unbounded fan-in and fan-out. This result had profound implications for areas such as learning theory, hardness of approximation, communication complexity, pseudorandomness, and circuit lower-bounds. For learning, specifically, \cite{linial1993constant} connected Fourier concentration to the learnability of circuit classes, resulting in the first quasi-polynomial sample- and time-complexity algorithm for PAC-learning any Boolean function implementable in \ACZ.

A natural quantum analog of \ACZ~is the \QACZ~circuit class, proposed by \cite{moore1999quantum}. Specifically, \QACZ~is the class of constant-depth (shallow) quantum circuits comprised of arbitrary single-qubit gates and polynomial $CZ_k$ gates of unbounded width $k$\footnote{The $CZ_k$ gate applies a $-1$ phase to a state if all $k$ inputs are in the $\ket{1}$ state and otherwise leaves the state unchanged. This is a natural quantum analog of the classical $AND$ operation.} (many-qubit gates). It is a long-standing open question as to whether the quantum fan-out operation can be implemented in \QACZ, i.e. whether \QACZ = \QACZf~(where the $f$ denotes ``with fan-out'')\footnote{Quantum fan-out is a many-qubit gate, with one control qubit that applies bit-flips to many target qubits. Since quantum fan-out is equivalent to parity via Hadamard conjugation, if fan-out $\in$ \QACZ,  this would imply parity $\in$ \QACZ. This is in contrast to \ACZ, which trivially contains fan-out, but does not contain parity \citep{linial1993constant}.}. If fan-out is in \QACZ, then powerful subroutines such as sorting, arithmetic operations, phase estimation, and the quantum Fourier transform could be approximately implemented in \QACZ~\citep{hoyer2005quantum}. Furthermore, \QACZ~contains \QNCZ~(the quantum analog of \NCZ), which is defined similarly to \QACZ, but constrained to constant-width $CZ_k$ gates, i.e. $k=O(1)$. Since constant-width gates are what we can feasibly implement on real quantum hardware and \QNCZ $\subset$ \QACZ $\subset$ \QNC$^1$, learning \QACZ~makes progress towards learning practical logarithmic-depth circuits.

Recent work has made exciting progress towards efficient learning of \QACZ. Notably, \cite{huang2024learning} established an efficient worst-case  algorithm for learning \QNCZ~unitaries. However, this algorithm relies crucially on the fact that \QNCZ~gates act only on a constant number of qubits, making it unclear how to generalize the approach to \QACZ~(with many-qubit gates). In another vein, \cite{nadimpalli2023pauli} proved a quantum analog of \cite{linial1993constant}'s celebrated Fourier-concentration result. In particular, they demonstrated that the Pauli spectrum\footnote{The Pauli spectrum is the quantum analog of the Fourier spectrum of classical Boolean functions. The notion of quantum Boolean functions was originally proposed by \cite{montanaro2010quantum}.} of single-output \QACZ~circuit channels (with limited ancillas) are low-degree concentrated. Interestingly, however, unlike the classical Fourier-concentration result of \cite{linial1993constant}, this does not straightforwardly imply an efficient learning algorithm for full characterization of \QACZ~circuits.

To see why, note that in the classical setting of \cite{linial1993constant}, although learned Fourier coefficients are not exact, there exists a simple projection from the learned monomial onto the space of valid Boolean functions.
This adds negligible time-complexity to the overall learning algorithm. However, in learning single-output channels of $n$-qubit \QACZ~circuits \citep{nadimpalli2023pauli}, the central object of study is an $O(4^{n})$-dimensional matrix (corresponding to the channel's Choi representation). Projecting the algorithm's learned matrix onto the space of valid Choi representations requires solving a semi-definite programming problem, which can require exponential time in the worst-case. Thus, the overall runtime of \cite{nadimpalli2023pauli}'s learning algorithm is far from efficient. Furthermore, even if the algorithm were time-efficient, it would only describe one of the \QACZ~circuit's $n$ outputs and is limited to \QACZ~circuits with at most logarithmic ancilla. 

\begin{table}[t!] 

\centering
\renewcommand{\arraystretch}{1.2} 
\setlength{\tabcolsep}{3pt}       
\footnotesize                            
\begin{tabular}{|c|c|c|c|c|c|c|c|c|}
\hline
\textbf{Circuit}  & \multirow{2}{*}{\textbf{Type}} & \multirow{2}{*}{\textbf{Fan-in}} & \multirow{2}{*}{\textbf{Paper}} &  \textbf{Learned} &\textbf{\#}  & \multicolumn{2}{c|}{\textbf{Complexity}} & \textbf{\#}\\ \cline{7-8}
\textbf{Class} & & & & \textbf{Object} & \textbf{Outputs}  & \textbf{Time} & \textbf{Sample}  & \textbf{Ancilla}\\ \hline
\textbf{AC$^0$} & Classical & $\infty$ & \cite{linial1993constant} & Monomial & 1  & $O(n^{\log^d n})$ & $O(n^{\log^d n})$  & -- \\ \hline
\textbf{QNC$^0$} & Quantum & $O(1)$ & \cite{huang2024learning} & Unitary & n  & $O(n)$ & $O(n)$ & $\infty$ \\ \hline
\multirow{2}{*}{\textbf{QAC$^0$}} & \multirow{2}{*}{Quantum} & \multirow{2}{*}{$\infty$} & \cite{nadimpalli2023pauli}  & Choi Rep & \cellcolor{red!20}  1  &\cellcolor{red!20} $O(2^n)$ &  \cellcolor{green!20}  & \cellcolor{red!20} \\ \cline{4-7} 
 & & & This work & Unitary & \cellcolor{green!20} n  & \cellcolor{green!20} $O(n^{\log^d n})$ & \multirow{-2}{*}{\cellcolor{green!20} $O(n^{\log^d n})$}  & \multirow{-2}{*}{\cellcolor{red!20} $O(\log n)$} \\ \hline
\end{tabular}
\caption{Comparison of learning algorithms for shallow (depth-$d$) classical and quantum circuits. For the quantum algorithms we also report the number of ancilla the guarantees hold for.}
\label{tab:circuit_classes}
\end{table}

In this paper, we address most limitations of the prior work, by proposing the first sample- and \emph{time-efficient} algorithm for learning the \emph{full unitary} implemented by a \QACZ~circuit. In particular, we show that $n$-qubit \QACZ~unitaries are learnable to an average-case distance, with high probability, in quasi-polynomial sample- and time-complexity. This matches both the sample- and time-complexity of \cite{linial1993constant}. Furthermore, we offer an efficient circuit synthesis procedure for our learned unitaries, making progress towards proper learning of \QACZ~circuits. The key remaining limitation, as in the prior work of \cite{nadimpalli2023pauli}, is that our learning algorithm's provable guarantees only extend to \QACZ~circuits with up to logarithmic ancilla. However, we offer a conjecture about the Pauli spectrum of \QACZ~circuits, which would extend our guarantees to \QACZ~circuits with polynomially many ancilla. \Cref{tab:circuit_classes} offers a summary of our algorithm's learning guarantees and comparison to prior work.

\section{Overview of Results and Contributions} \label{sec:results}
The primary contribution of this work is the first sample- and time-efficient algorithm for learning shallow quantum circuits with many-qubit gates. Namely, we prove the following main result:

\begin{theorem}[Efficient learning of $n$-output \QACZ~unitaries] \label{thm:main_result_intro}
    Consider an unknown $n$-qubit, depth-$d$ \QACZ~circuit, implementing unitary $C$, and failure probability $\delta \in (0,1)$. With high probability, $1-\delta$, we can learn a $2n$-qubit unitary $U$, such that the average gate fidelity\footnote{As discussed in \Cref{sec:dist_meas}, the average gate fidelity is a measure of how distinct two unitary channel outputs are, for pure input states averaged over the Haar (uniform) measure.} $\davg$ is bounded as

    \begin{align}
        \davg(U, C \otimes C^\dagger) \leq 1/\poly(n).
    \end{align}

    \noindent Moreover, $U$ is learnable with quasi-polynomial, $O(n^{\log^d n}\log(1/\delta))$, sample- and time-complexity.
\end{theorem}

\noindent Note that our procedure makes use of classical shadow tomography \citep{huang2020predicting} and, thus, falls within the standard quantum-PAC learning model. However, we use shadow tomography to learn the Heisenberg-evolved single-qubit Pauli observables of $U$, meaning we require query access to both $U$ and $U^\dagger$, with stabilizer states as input and Pauli-basis output measurements.

We also prove the following unitary synthesis result, which makes progress towards a quasi-polynomial time proper learning algorithm for \QACZ~circuits. Namely, we show that the learned unitary can be synthesized in \QAC, or the class of all polylogarithmic-depth quantum circuits consisting of arbitrary single-qubit gates and many-qubit $CZ$ gates.

\begin{proposition}[Learning \QAC~implementations of \QACZ] \label{thm:uni_synthesis_intro}
    Given the learned unitary $U$ which is $1/\poly(n)$-close to \QACZ~circuit $C$, there exists a quasi-polynomial time algorithm to learn a \QAC~circuit implementing unitary $U^*$ such that $\davg(U^*,C \otimes C^\dagger) \leq 1 /\poly(n)$.  
\end{proposition}

Our work builds upon and extends the results of \cite{nadimpalli2023pauli} and \cite{huang2024learning}. While \cite{nadimpalli2023pauli} studied the Choi representations of single-output \QACZ~channels, \cite{huang2024learning} focused on learning and ``sewing'' Heisenberg-evolved single-qubit Pauli observables of \QNCZ~circuits. To combine these techniques, we establish a connection between the Choi representations of single-output \QACZ~channels and Heisenberg-evolved single-qubit Pauli observables of \QACZ~circuits. We re-prove several key results from \cite{nadimpalli2023pauli} for \QACZ~Heisenberg-evolved single-qubit Pauli observables, establishing their low-degree concentration. Furthermore, we make the key observation that \QACZ~observables are not just low-degree concentrated, but are in fact low-support concentrated---i.e. they are ($\poly\log n$)-juntas, with non-zero Pauli terms supported on only $\poly \log n$ of the $n$ qubits. This strengthened concentration result is critical in achieving the accuracy guarantees and efficient computational complexity of our learning algorithm. Overall, we believe that these concentration results for Heisenberg-evolved Pauli observables are of independent interest beyond the learning algorithm studied in this work.

Our algorithm improves upon both \cite{huang2024learning} and \cite{nadimpalli2023pauli} in several key ways. Foremost, we successfully generalize the learning procedure of \cite{huang2024learning} to circuits with \emph{many-qubit} gates, a setting to which their proof techniques do not directly apply. This generalization is made possible by the new observable concentration results, which show that learning observables with \emph{restricted support} suffices to approximate \QACZ~observables. The restricted support of these learned observables is crucial in obtaining the learning algorithm's efficient sample and time complexity. In contrast to \cite{nadimpalli2023pauli}, which focuses on learning the low-degree approximation of the \emph{Choi representation} of a \emph{single-output} \QACZ~channel (necessitating the tracing out of $n-1$ qubits), our algorithm accomplishes efficient learning of the \emph{unitary} of the \emph{entire $n$-output} \QACZ~circuit. Most importantly, our algorithm achieves quasi-polynomial sample and time complexity, whereas \cite{nadimpalli2023pauli} attains quasi-polynomial sample complexity, but suffers an exponential time complexity.

While our algorithm represents a significant advancement in learning \QACZ~circuits, it has two key limitations compared to previous work. A notable drawback of our learning algorithm compared to that of \cite{huang2024learning} is that ours only offers an average-case learning guarantee, not a worst-case guarantee\footnote{\cite{huang2024learning}'s worst-case guarantee implies accurate outputs for \emph{any} input. In contrast, average-case learning guarantees accurate outputs for \emph{most} inputs.}. However, we observe that the \cite[Proposition 3]{huang2024learning} lower-bound for hardness of learning logarithmic-depth quantum circuits to diamond-norm distance also applies to \QACZ~circuits. Thus, we justify our use of an average-case measure by showing that no algorithm can efficiently worst-case learn arbitrary \QACZ~circuits:

\begin{proposition}[Hardness of learning \QACZ] \label{thm:hardness_intro}
Learning an unknown $n$-qubit unitary $U$, generated by a \QACZ~circuit, to $\frac{1}{3}$ dimond distance with high probability requires $\exp(\Omega(n))$ queries.
\end{proposition}

\noindent Furthermore, it was not immediately obvious that the worst-case learning guarantees of \cite{huang2024learning} would translate to the average-case setting for \QACZ. The proof of these average-case guarantees was a substantial and notable contribution of our work, elaborated in \Cref{sec:learn}.

Another limitation of our learning algorithm, compared to \cite{huang2024learning}, is that it only works for \QACZ~circuits with a logarithmic ancilla qubits. Note, however, that this ancilla constraint matches that of the learning algorithm in \cite[Theorem 39]{nadimpalli2023pauli}. This constraint arises from the nature of many-qubit \QACZ~gates, which allow light-cones to rapidly encompass a large number of ancilla qubits, even in constant depth. In contrast, the constant-width gate setting of \cite{huang2024learning} ensures that at most a constant number of ancillas are in any output qubit's light-cone, implying that most ancilla qubits do not affect the computation.
Our algorithm's reliance on light-cone style arguments about the support of Heisenberg-evolved observables necessitates restricting the number of ancilla qubits to prevent error blow-up. Moreover, as noted in \cite{nadimpalli2023pauli}, the pre-specified input value of ancilla qubits poses a challenge for average-case arguments and removal of large-$CZ_k$ gates.
To make progress towards the many-ancilla case, we show that if the following conjecture, which is a strengthening of \cite[Conjecture 1]{nadimpalli2023pauli}, is proven true, then our learning procedure for pollynomially-many ancilla qubits:

\begin{conjecture}[Ancilla-independent low-support concentration] \label{conj:anc_conc}   Suppose $C$ is a depth-$d$, size-$s$ \QACZ~circuit performing clean computation $C (I \otimes \ket{0^a}) = A \otimes \ket{0^a}$ on $n+a$ qubits, for $\poly(n)$ ancilla $a$. Let $\obsared = A (P_i \otimes I^a) A^\dagger$ be a Heisenberg-evolved single-qubit Pauli observable with ancilla restriction.
For support $\calS$ such that $|\calS|=k^d$, the weight outside $\calS$ is bounded as
    \begin{align}
        \wt^{\notin\calS}[\obsared] \leq \poly(s) \cdot 2^{-\Omega(k^{1/d})}.
    \end{align}
\end{conjecture}

\begin{corollary}[Learning \QACZ~with polynomial ancilla] \label{thm:poly_anc}
    Assume Conjecture~\ref{conj:anc_conc} holds.
    With quasi-polynomial samples and time, we can learn $2n$-qubit unitary $A_\text{sew}$, with high probability, such that
    \begin{align}
        \davg(A_\text{sew},A \otimes A^\dagger) \leq 1 / \poly (n).
    \end{align}
\end{corollary}

\noindent This corollary suggests a potential pathway for handling a broader class of \QACZ~circuits.

\section{Prior Work}
Sample-efficient learning algorithms have been developed for a wide range of quantum processes. For example, studies of generalization in quantum machine learning models have yielded sample-efficient algorithms for average-case learning of polynomial-size quantum circuit unitaries \citep{caro2022generalization, caro2023out, zhao2023learning}. Additionally, \cite{huang2022quantum} established sample-efficient average-case learning algorithms for polynomial-complexity channels. These results leverage techniques such as shadow tomography \citep{aaronson2018shadow, buadescu2021improved, huang2020predicting} to achieve sample efficiency. However, while the algorithms are sample-efficient, they often require predicting exponentially many observables and, thus, are not computationally efficient for the full learning task. To understand whether this is a limitation of all algorithms, \cite{zhao2023learning} established quantum computational hardness for average-case learning of polynomial-sized quantum circuits, assuming quantum hardness of learning with errors \citep{regev2009lattices}.

Computational efficiency necessitates focusing on specific prediction tasks or more structured families of unitaries and channels. For instance, \cite{huang2023learning} and \cite{chen2024predicting} provided a quasi-polynomial time average-case learning algorithm for accurate prediction of arbitrary $n$-qubit channels. These algorithms are efficient when predicting a given observable on the output of the unknown quantum channel, for most input states.
Restricting to certain families of unitaries and channels is another common strategy to attain computational efficiency. Computationally-efficient learning algorithms have been developed for various families such as unitaries generated by short-time Hamiltonian dynamics \citep{yu2023robust, huang2023learning, bakshi2024structure, haah2024learning, stilck2024efficient}, Pauli channels under local or sparsity constraints \citep{flammia2020efficient, flammia2021pauli, chen2022quantum, chen2023learnability, caro2024learning, chen2024tight}, Clifford circuits with limited non-Clifford gates \citep{lai2022learning, grewal2023efficient, leone2024learning, grewal2024improved, du2024efficient}, and quantum juntas \citep{chen2023testing, bao2023nearly}. Especially relevant to our work are \citep{nadimpalli2023pauli}'s algorithm for learning single-output \QACZ~channels and \cite{huang2024learning}'s algorithm for learning shallow quantum circuits comprised of constant-width gates (\QNCZ), which we now describe in more detail.

\paragraph{Learning shallow quantum circuits.}
\cite{huang2024learning} gave the first polynomial sample- and time-complexity algorithm for learning unitaries of shallow quantum circuits with constant-width gates (\QNCZ). The key insight was that, although these circuits can generate highly non-local and classically-hard output distributions, the associated unitaries can be efficiently reconstructed (to high-accuracy) via the local light-cones of each output qubit. Namely, for an $n$-qubit shallow circuit $U$, they developed an efficient algorithm for learning the $3n$ Heisenberg-evolved single-qubit Pauli observables $O_{P_i} = U P_i U^\dagger$. They proved that the learned $\widetilde{O}_{P_i}$ observables are close to the true $O_{P_i}$ observables, with respect to the operator norm, and can be efficiently ``sewed'' into unitary
\begin{align}
    U_\text{sew} = \textnormal{SWAP}^{\otimes n} \prod_{i=1}^n \left[\proj_\infty\left(\frac{1}{2} I\otimes I+\frac{1}{2}\sum_{P\in\{X,Y,Z\}} \widetilde{O}_{P_i} \otimes P_i\right) \right] ,
\end{align}
where $\proj_\infty$ is the projection onto the unitary minimizing the operator norm distance with respect to the input matrix and $\textnormal{SWAP}^{\otimes n}$ is the SWAP operation between the top and bottom $n$ qubits. Finally, they proved that sewn unitary $ U_\text{sew}$ is close to the true unitary $U \otimes U^\dagger$ under diamond distance. Note that the algorithm's sample and time efficiency rely heavily on the constant-sized support of observables $O_{P_i}$ and  $\widetilde{O}_{P_i}$. However, \QACZ~observables can have polynomial-sized support.

\paragraph{On the Pauli Spectrum of \QACZ.}
\cite{nadimpalli2023pauli} achieved a quantum analog of \cite{linial1993constant}'s Fourier concentration bound for \ACZ. For an $n$-qubit \QACZ~circuit implementing unitary $C$, they defined single-qubit output channel $\calE_C (\rho) = \Tr_{n-1} \left(C \rho C^\dagger\right)$,
where $\rho$ is an $n$-qubit density matrix and $\Tr_{n-1}$ denotes the partial trace over $n-1$ of the qubits, leaving a singular designated output qubit. Via the standard Choi-Jamiołkowski isomorphism, they mapped this channel into a Choi representation, decomposable into the Pauli basis as 
\begin{align}
    \Phi_{\calE_C} = \left(I^{\otimes n} \otimes \calE_C \right) \ketbra{\text{EPR}_n}{\text{EPR}_n}= \sum_{P \in \{I,X,Y,Z\}^{\otimes n}} \widehat{\Phi}_{\calE_C}(P) \cdot P,
\end{align}
where $\ket{\text{EPR}_n} = \sum_{x \in [n]} \ket{x} \otimes \ket{x}$ denotes the un-normalized Bell state and  $\widehat{\Phi}_{\calE_C}(P)$ are the Pauli (Fourier) coefficients\footnote{For more background on Pauli analysis, we refer the reader to \Cref{sec:prelims}.}. They proved that $\Phi_{\calE_C}$ is low-degree concentrated by showing that the Pauli weight for all Paulis $P$ of degree $>k$, denoted $\wt^{>k}\left[\Phi_{\calE_C}\right]$, decays exponentially in $k$. Formally, if size $s$ denotes the total number of $CZ_k$ gates in the circuit, $|P|$ denotes the degree\footnote{The degree of a Pauli string is the number of qubits upon which the Pauli operator $P$ acts non-trivially.}, and $d=\mathcal{O}(1)$ denotes the circuit depth,
\begin{align}
    \wt^{>k}\left[\Phi_{\calE_C}\right]=\sum_{|P| > k} \left|\widehat{\Phi}_{\calE_C}(P)\right|^2 \leq \calO \left(\frac{s^2}{2^{k^{1/d}}}\right).
\end{align}

Since  $\Phi_{\calE_C}$ is low-degree concentrated, \cite{nadimpalli2023pauli} learned $\Phi_{\calE_C}$'s low-degree Pauli coefficients, i.e. $\widetilde{\Phi}_{\calE_C}(P)$ for all $P \in \calF = \{Q \in \pauli^n: |Q| < \poly\log(n)\}$, to construct approximate Choi representation $\widetilde{\Phi}_{\calE_C} = \sum_{P \in \calF} \widetilde{\Phi}_{\calE_C}(P) \cdot P$ and proved that learned matrix $\widetilde{\Phi}_{\calE_C}$ is close to the true Choi representation $\Phi_{\calE_C}$ (according to the normalized Frobenius distance). This resulted in the first quasi-polynomial sample-complexity algorithm for learning single-output \QACZ channels. However, since projecting the learned matrix onto a valid Choi representation requires solving a exponential-dimensional semi-definite program, the runtime is exponential.

\section{Proof Overview}

At a high level, our algorithm for learning $n$-qubit \QACZ~circuits is a nontrivial synthesis of \cite{nadimpalli2023pauli}'s \QACZ~concetration results and \cite{huang2024learning}'s sewing technique. Our main technical contribution is bridging the high-level concepts introduced in the two works and developing an efficient end-to-end learning procedure with provable guarantees. We will now provide a high-level overview of and intuition for the key results and techniques introduced in this work. 

\paragraph{Choi Representations and Heisenberg-Evolved Observables.}
While \cite{nadimpalli2023pauli} proved low-degree concentration of the \emph{Choi representation} of single-output channels, the \cite{huang2024learning} learned low-degree approximations of single-qubit Heisenberg-evolved Pauli \emph{observables}. In \Cref{sec:choi_obs_relate}, we build a translator between these two quantum primitives for \QACZ. To summarize, while the Choi representation encodes information about the \QACZ~channel with respect to any input and measurement observable, the Heisenberg-evolved observable restricts the channel with respect to a single measurement observable. Mathematically, for a \QACZ~circuit $C$, with single-output channel $\calE_C$, Choi representation $\Phi_{\calE_C}$, and single-qubit measurement observable $O$, the two representations are related via the dual-channel as
\begin{align} 
    \calE_C^* (O)= C^\dag(I_{n-1}\otimes O) C=\Tr_\text{out}\left((O_\text{out}\otimes I_\text{in}) \Phi_{\calE_C}\right)^\top.
\end{align}

\paragraph{Concentration of \QACZ~Heisenberg-Evolved Observables.}
Either the Choi representation or Heisenberg-evolved observable perspective could have been suitable for our work, but we chose observables due to their simpler mathematical nature and because learning requires circuit measurements with respect to select observables. This choice enables straightforward use of \cite{huang2024learning}'s Heisenberg-evolved observable sewing procedure, but means some translational work is necessary to leverage \cite{nadimpalli2023pauli}'s Choi representation concentration result. Thus, the first contribution of our work is a proof similar to \cite[Theorem 21]{nadimpalli2023pauli}, establishing low-degree concentration of \QACZ~Heisenberg-evolved single-qubit Pauli observables. 
\begin{proposition}[Low-degree concentration] \label{thm:low_deg_conc_intro}
    For depth-$d$, size-$s$ \QACZ~circuit $C$ acting on $n$ qubits and Heisenberg-evolved single-qubit Pauli observable $O_{P_i} = C^\dagger P_i C$, for every degree $k \in [n]$:
    \begin{align}
        \wt^{>k}[O_{P_i}] = \sum_{|Q|>k} |\widehat{O}_{P_i}(Q)|^2 \leq \calO \left(s^2 2^{-k^{1/d}}\right).
    \end{align}
\end{proposition}
Conceptually, the proof of this result consists of two key steps. First, we show that if the \QACZ~circuit has no $CZ$ gates of width greater than $k^{1/d}$, then the Heisenberg-evolved observable's weight for degree $>k$ is zero. Then, we show that removing these ``large'' $CZ$ gates does not significantly change the Heisenberg-evolved observable under the normalized Frobenius distance measure, i.e.
\begin{lemma}[Large $CZ_k$ removal error] \label{thm:avg_dist_intro}
    For \QACZ~circuit $C$ and \QACZ~circuit $\widetilde{C}$, which is simply $C$ with all $m$ $CZ_k$'s of size $k> \gaterem$ removed, let $O_{P_i}=C P_i C^\dagger$ and $O^*_{P_i} = \widetilde{C} P_i \widetilde{C}^\dagger$ be the Heisenberg-evolved single-qubit Pauli observables. The average-case distance between these observables is
    \begin{align} 
        \frac{1}{2^n} \left\| O_{P_i}-O_{P_i}^*\right\|_F^2 \leq \epsilon^* = \frac{9m^2}{2^{\gaterem}}.
    \end{align}
\end{lemma}
\noindent Note that this distance bound is critical in proving the algorithm's learning guarantees. Furthermore, as will be discussed shortly, the size $\gaterem$ (of the smallest $CZ_k$ removed) must be carefully selected in order to learn the observable to high-precision, yet efficiently. Beyond showing that $O_{P_i}$ is low-degree concentrated, we also show that it is low-support, by adapting the proof of \Cref{thm:low_deg_conc_intro}.

\begin{lemma}[Low-support concentration] \label{thm:weight_ub_support_intro}
    For $\epsilon^*$ the same as in \Cref{thm:avg_dist_intro} and $\calS^* = \supp(O^*_{P_i})$, the weight of $O_{P_i}$ outside the support of $O_{P_i}^*$ is upper-bounded as
    \begin{align}
        \wt^{\notin\calS^*}[O_{P_i}] = \sum_{Q \in \calS^*} |\widehat{O}_{P_i}(Q)|^2
        \leq \frac{1}{2^n} \left\| O_{P_i}-O^*_{P_i}\right\|_F^2 = \epsilon^*.
    \end{align}
\end{lemma}
\noindent For more details about these concentration results and proofs, we refer the reader to \Cref{sec:heis_evolve_obs}.

\paragraph{Efficient Learning of \QACZ~Heisenberg-Evolved Observables.} Although efficient, \cite{huang2024learning}'s  algorithm for learning \QNCZ~Heisenberg-evolved observables relies crucially on the fact that these observables consist only of constant-width gates and, thus, have constant support. In the \QACZ~setting, however, we only have query access to $O_{P_i}$, which can have $\calO(n)$ support. Thus, we cannot directly use \cite{huang2024learning}'s algorithm to learn \QACZ~Heisenberg-evolved observables.

Meanwhile, \cite{nadimpalli2023pauli} proposed an algorithm with efficient sample complexity for learning approximate Choi representations of single-output \QACZ~channels. The algorithm leverages low-degree concentration of the \QACZ~Choi representation to efficiently approximate it, by only learning low-degree Pauli coefficients. Given our low-degree concentration result for \QACZ~Heisenberg-evolved observables (\Cref{thm:low_deg_conc}), one might naturally assume we could simply learn the approximate low-degree observable $\widetilde{O}_{P_i} = \sum_{|Q|\leq \ell} \widetilde{O}_{P_i}(Q) \cdot Q$, where $\widetilde{O}_{P_i}(Q)$ are the learned Pauli coefficients. However, as we will elaborate shortly, it is \emph{crucial} that the learned observables have \emph{low-support}, so as to ensure that the unitary projection step in the eventual observable sewing procedure can be computed efficiently. Therefore, we need a more sophisticated learning procedure that learns a low-support observable.

\begin{algorithm}[t]
\small
\caption{Heisenberg-Evolved Observable Learning Procedure} \label{alg:learning_alg}
\begin{enumerate}
    \item Learn all the Pauli coefficients of $O_{P_i}$ of degree $\leq \support$ to precision $\eta$, as in \Cref{eqn:d_bound}.
    \item Find the subset of Paulis supported on $\support$ qubits, $\calF_{\{s\}} \in \mathfrak{F}_\support$, with max weight amongst the learned coefficients, 
    \begin{align} \label{eqn:max_supp}
        T_\support = \underset{  \calF_{\{s\}} \in \mathfrak{F}_\support}{\arg\max} \sum_{Q \in \calF_{\{s\}}} \left|\widetilde{O}_{P_i}(Q)\right|^2.
    \end{align}
    \item Set all coefficients outside of this maximal-weight support to zero,
    \begin{align}
        \widetilde{O}^{(\support)}_{P_i}(Q) = \begin{cases}
            \widetilde{O}_{P_i}(Q), & \text{ if } Q \in T_\support, \\
            0, & \text{ otherwise.}
        \end{cases}
    \end{align}
    \item Output the learned observable, which is fully supported on only $\support$ qubits,
    \begin{align} \label{eqn:final_learned}
        \widetilde{O}^{(\support)}_{P_i} = \sum_{Q \in T_\support} \widetilde{O}^{(\support)}_{P_i}(Q) \cdot Q.
    \end{align}
\end{enumerate}

\end{algorithm}

In light of these difficulties, we propose a new learning algorithm (\Cref{alg:learning_alg}) that efficiently learns approximate \QACZ~Heisenberg-evolved Pauli observables $\widetilde{O}^{(\support)}_{P_i}$ with guaranteed low-support~$\support$. Similar to the algorithm of \cite{nadimpalli2023pauli}, this algorithm uses classical shadow tomography \cite{huang2020predicting} to approximately learn all Pauli coefficients $\widetilde{O}_{P_i} (Q)$ for all Paulis $|Q|\leq \support$. However, after all these coefficients are learned, the Pauli coeffcients are grouped into sets of $\support$-qubit support, $\calF_{\{s\}} \in \mathfrak{F}_\support$. All coefficients which lie outside the set $T_\support$ of maximal weight, as expressed in \Cref{eqn:max_supp}, are set to zero. This enforces that the final learned observable $\widetilde{O}_{P_i}^{(\support)}$, with decomposition given in \Cref{eqn:final_learned}, is supported on only $\support$ qubits.

Leveraging our low-support concentration result (\Cref{thm:weight_ub_support_intro}),  we bound the Frobenius distance between $\widetilde{O}_{P_i}^{(\support)}$ and $O_{P_i}$ according to the learning accuracy and distance between $O_{P_i}$ and $O_{P_i}^*$.
\begin{lemma}[\Cref{alg:learning_alg} error bound] \label{thm:learn_dist_intro}
    Let $O_{P_i}$ be a \QACZ~Heisenberg-evolved Pauli observable, which is $\epsilon^*$-close to the observable $O_{P_i}^*$, with all gates of width $\geq \gaterem$ removed. Furthermore, suppose that we can learn all the  degree-$\support$ Pauli coefficients of $O_{P_i}$ to precision $\eta$, i.e.
    \begin{align} \label{eqn:d_bound}
        \left| \widehat{O}_{P_i}(Q)-\widetilde{O}_{P_i}(Q)\right| \leq \eta, \quad \forall Q \in \{P \in \pauli^n: |P|\leq \support\}.
    \end{align}
    Leveraging these learned coefficients, \Cref{alg:learning_alg} will produce learned observable $\widetilde{O}^{(\support)}_{P_i}$, such that
    \begin{align} \label{eqn:e_bound}
        \frac{1}{2^n} \left\| \widetilde{O}^{(\support)}_{P_i} - O_{P_i} \right\|_F^2=\epsilon_{P_i} \leq 2 \cdot  4^\support \cdot \eta^2 + \epsilon^*
    \end{align}
\end{lemma}
\noindent Therefore, leveraging \Cref{thm:avg_dist_intro} and setting $\gaterem = \calO(\log n)$ (which implies 
 $\support = \calO(\log^d n)$), we establish that, in quasi-polynomial sample- and time-complexity, we can learn an approximate observable $\widetilde{O}_{P_i}^{(\support)}$, supported on $\calO(\log^d n)$ qubits, which is $1/\poly(n)$-close to $O_{P_i}$. \Cref{fig:obs_relate} offers a conceptual illustration of the key observables involved in our algorithm and their relations. For a detailed description of these learning results and proofs, we refer the reader to \Cref{sec:learn_obs}.

\begin{figure}[t]

    \centering
    \includegraphics[width=0.7\linewidth]{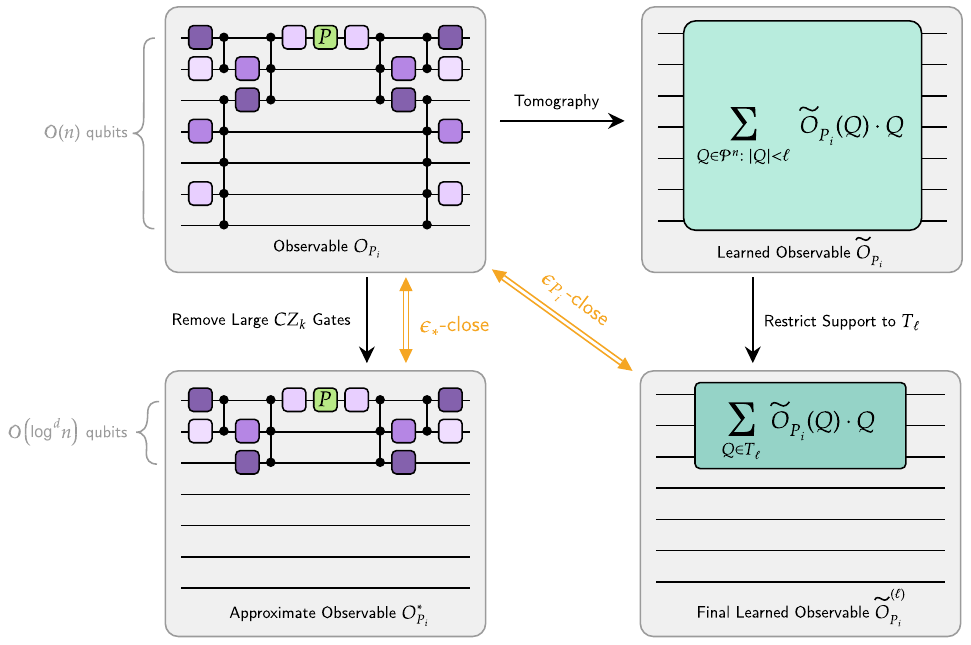}
    \caption{A conceptual illustration of all the observables involved in the \QACZ~Heisenberg-Evolved Learning procedure and their relations.}
    \label{fig:obs_relate}

\end{figure}

\paragraph{Efficient Sewing of \QACZ~Heisenberg-Evolved Observables.} This procedure is used to efficiently learn all $3n$ of the \QACZ~Heinsenberg-evolved single-qubit Pauli observables. \cite{huang2024learning}'s light-cone sewing procedure (with a modified projection operator) is then used to sew all $3n$ learned observables $\widetilde{O}^{(\support)}_{P_i}$ into approximate unitary $C_\text{sew}$, close to $C \otimes C^\dagger$ in average-case distance. 
\begin{lemma}[Sewing error bound] \label{thm:csew_bound_intro}
    Suppose $C$ is an $n$-qubit \QACZ~circuit, which has a set of Heisenberg-evolved observables $\left\{O_{P_i}\right\}_{i,P}$, corresponding to each of the $3n$ single-qubit Paulis $P_i$. Let $\left\{\widetilde{O}^{(\support)}_{P_i}\right\}_{i,P}$ denote the set of learned observables, which are at most $\epsilon_{P_i}$-far from the true observables. By ``sewing'' the learned observables, construct the unitary 
    \begin{align} \label{eqn:csew_unitary}
        C_\text{sew} \left(\{\widetilde{O}_{P_i}^{(\support)}\}_{i,P}\right) := \textnormal{SWAP}^{\otimes n} \prod_{i=1}^n \left[\proj_U\left(\frac{1}{2} I\otimes I+\frac{1}{2}\sum_{P\in\{X,Y,Z\}} \widetilde{O}_{P_i}^{(\support)} \otimes P_i\right) \right],
    \end{align}
    where $\proj_U$ is the projection onto the unitary minimizing the Frobenius norm distance and $\textnormal{SWAP}^{\otimes n}$ swaps the first and last $n$ qubits. The average-case distance between $C_\text{sew}$ and $C \otimes C^\dagger$ is at most
    \begin{align}
        \davg(C_\text{sew},C \otimes C^\dagger) \leq \frac{1}{2}~\sum_{i=1}^n \sum_{P\in\{X,Y,Z\}} \epsilon_{P_i}.
    \end{align}
\end{lemma}

Since our result leverages average-case distance measures, we had to re-prove most of the learning guarantees of \cite{huang2024learning}, which were only proven for worst-case distance measures. Furthermore, to ensure that $C_\text{sew}$ is unitary, the average-case version of the sewing procedure leverages the $\proj_U$ projection onto the unitary minimizing the Frobenius norm distance. Solving for this unitary is the solution to the well-established Orthogonal Procrustes Problem, which requires computing a singular value decomposition (SVD) and setting all singular values to one. In general, computing the SVD of a $2n$-qubit unitary requires exponential time. However, we imposed that the learned observables $\widetilde{O}_{P_i}^{(\support)}$ are supported on only $\calO(\log^d n)$ qubits. This implies that each matrix, $\frac{1}{2} I\otimes I+\frac{1}{2}\sum_{P\in\{X,Y,Z\}} \widetilde{O}_{P_i}^{(\support)} \otimes P_i$, in the projection also has $\calO(\log^d n)$ support. Since we only need to compute the SVD of the non-trivial subsystem within the support, the computational complexity of the sewing procedure is only quasi-polynomial, i.e. $\calO(2^{\poly\log n})$. This is the central reason that we impose that the learned observables are not only low-degree, but also low-support.

By plugging our $1/\poly(n)$ error bound for $\epsilon_{P_i}$ into \Cref{thm:csew_bound_intro}, as well as summing the sample and time complexities of the learning and sewing procedures, we can guarantee that our algorithm efficiently learns $n$-output \QACZ~unitaries, thereby proving \Cref{thm:main_result_intro}. For a more detailed description of these sewing results and proofs, we refer the readers to \Cref{sec:sew}.

\paragraph{Learning \QAC~Circuits with Improved Depth.} With our efficient procedure for learning the unitary corresponding to an $n$-output \QACZ~circuit, one could ask whether we could learn a \QACZ~circuit synthesizing this unitary, also known as ``proper learning'' of \QACZ. In this work, while we do not achieve a proper learning algorithm for \QACZ, we do make progress towards one. A naive attempt at implementing the learned \QACZ~unitary would require a \QAC~circuit of quasi-polynomial depth. However, in \Cref{thm:uni_synthesis_intro} we establish that the circuit depth can be reduced to poly-logarithmic.

Begin by noticing that the learned \QACZ~unitary $C_\text{sew}$, as given in \Cref{eqn:csew_unitary}, sews the learned unitaries $W_i = \proj_U\left(\frac{1}{2} I\otimes I+\frac{1}{2}\sum_{P\in\{X,Y,Z\}} \widetilde{O}_{P_i}^{(\support)} \otimes P_i\right)$
via SWAP gate operations, which are implementable in \QACZ. To achieve a poly-logarithmic depth \QAC~implementation of $C_\text{sew}$, we prove the following two key results: 1) each unitary $W_i$ can be implemented to high accuracy via a \QACZ~circuit and 2) there exists a sewing order with parallelization such that all $n$ of the $W_i$ unitaries can be implemented and sewed in worst-case poly-logarithmic depth.

To prove the first result, we decompose the \QACZ~circuit structure into two key parts: the $CZ_k$ gates and the arbitrary single-qubit rotation gates. Since $W_i$ is supported on only $\calO(\log^d n)$ qubits, a \QACZ~circuit operating on $\calO(\log^d n)$ qubits has at most quasi-polynomial configurations of $CZ_k$ gates with $2\leq k \leq \calO(\log n)$. For each such configuration, we show that it is possible to construct a polynomial sized  $1/\poly(n)$-net over the polynomial number of arbitrary single-qubit rotation gates. Therefore the $1/\poly(n)$-net over all such \QACZ~circuit architectures has quasi-polynomial size and a simple brute-force search can be leveraged to find the circuit closest to $W_i$.

To prove the second result, demonstrating that the learned \QACZ~circuits can be sewn into a poly-logarithmic depth \QAC~circuit, we leverage a graph-coloring argument similar to \cite[Lemma 13]{huang2024learning}. Note that the poly-logarithmic depth arises due to the poly-logarithmic support of the $W_i$ unitaries, resulting in a poly-logarithmic degree and, thus, worst-case coloring of the graph. For more details about this unitary synthesis procedure, we refer the reader to \Cref{sec:improve_depth}.

\paragraph{Concentration and Learning of \QACZ~with Ancilla.} In \Cref{sec:ancillas}, we discuss the applicability of our learning algorithm to \QACZ~circuits with ancilla. Similar to \cite{huang2024learning}, we only consider circuits where the ancillas are initialized to the $\ket{0^a}$ state and the computation is clean (meaning ancillas are reverted to $\ket{0^a}$ by the end of the computation). This implies that the action of $C$ on the $(n+a)$-qubit system is equivalent to the action of another unitary $A$ on just the $n$-qubit system without ancillas, i.e. $C (I \otimes \ket{0^a}) = A \otimes \ket{0^a}$. Thus, we define the Heisenberg-evolved Pauli observables of this system ``without ancilla restriction'' as $\obsanc = C (P_i \otimes I^a) C^\dagger $
and ``with ancilla restriction'' as $\obsared = (I \otimes \bra{0^a}) \cdot \obsanc \cdot (I \otimes \ket{0^a}) = A P_i A^\dagger$. Via Pauli analysis, we prove that the Pauli weight spectrums of $\obsanc$ and $\obsared$ are related by a $2^a$ multiplicative factor or, in other words, an exponential blow-up factor in the number of ancilla. 
\begin{lemma}[Effect of ancilla]  Let $\calS \subseteq \pauli^n$ be a subset of the set of $n$-qubit Paulis, then 
    \begin{align}
        \wt^{\in \calS}[\obsared] \leq 2^a \cdot \wt^{\in\calS}[\obsanc].
    \end{align}
\end{lemma}
\noindent We also show that the distance between $\obsared$ and $\obsared^*$ (the circuit with large $CZ$ gates removed) can be bounded by the distance between $\obsanc$ and $\obsanc^*$, but that a $2^a$ multiplicative blow-up factor also arises relative to the distance in the ancilla-free case.
\begin{lemma}[Large $CZ_k$ removal error with ancillas] 
    The average-case distance between observables of \QACZ~circuit $C$ and \QACZ~circuit $\widetilde{C}$ ($C$ with all $m$ $CZ_k$'s of size $k> \gaterem$ removed) satisfies
    \begin{align}
        \frac{1}{2^n} \left\| \obsared-\obsared^*\right\|_F^2\leq \frac{1}{2^n} \| \obsanc - \obsanc^* \|^2_F \leq 2^a \cdot \frac{9m^2}{2^{\gaterem}}.
    \end{align}
\end{lemma}
Therefore, by using \Cref{alg:learning_alg} to learn the $n$-qubit observables with ancilla restriction,  
\begin{align} \label{eqn:dist_anc_bound}
    \frac{1}{2^n} \left\| \widetilde{O}^{(\support)}_{P_i,n} - \obsared \right\|_F^2 \leq 2 \cdot  4^\support \cdot \eta^2 + \wt^{\notin\calS^*}[\obsared] \leq 2 \cdot 4^\support \cdot \eta^2 + 2^a \cdot \frac{9m^2}{2^{\gaterem}}
\end{align}
where $\eta$ is the learning accuracy. To retain our algorithm's efficient sample and time complexity, we maintain that $\gaterem\leq \calO(\log(n+a)) \leq \calO(\log n)$. However, due to the $2^a$ factor in the numerator of the $\wt^{\notin \calS^*}[O_{P_i}]$ term, the distance is guaranteed to be $\leq1/\poly(n)$ only if the number of ancillas is logarithmic, i.e. $a = \calO(\log n)$. Thus, we prove that our learning procedure holds for \QACZ~circuits with a logarithmic number of ancilla.

However, if the $2^a$ term was eliminated from the upper-bound on $\wt^{\notin\calS^*}[\obsared]$ in \Cref{eqn:dist_anc_bound}, then the distance would be  $\leq1/\poly(n)$ and $\gaterem\leq \calO(\log(n+a)) \leq \calO(\log n)$ would be satisfied for a number of ancillas $a$ polynomial in $n$. This observation motivates \Cref{conj:anc_conc}, a strengthening of \cite[Conjecture 1]{nadimpalli2023pauli}. If proven true, \Cref{thm:poly_anc} implies a quasi-polynomial algorithm for learning unitaries of \QACZ~circuits with polynomially many ancillas.

\paragraph{Hardness of Learning~\QACZ.} We conclude the work by proving \Cref{thm:hardness_intro}, as discussed in \Cref{sec:results} and \Cref{sec:hardness_proof}, which demonstrates that \QACZ~cannot be efficiently learned according to the diamond-norm distance measure -- motivating our use of an average-case measure.

\section{Discussion and Future Work}
There are several interesting directions for future work that would strengthen our understanding of \QACZ~and improve the applicability of this learning procedure.

\paragraph{Ancilla.} Although our learning procedure handles  logarithmic ancilla, many quantum algorithms require far more ancilla to, e.g., perform block-encodings and error correction. We proved that our algorithm works for polynomial ancilla under \Cref{conj:anc_conc}, but proving this conjecture remains an important and challenging open problem. Concurrent work by \cite{anshu2024computational} proved that \QACZ~projectors are well-approximated by low-degree projectors according to the spectral norm, with guarantees for slightly superlinear ancilla, but it is not clear if this result extends to our setting. 
Furthermore, in this work, we prove algorithmic guarantees for \QACZ~circuits that perform clean computation on their ancillas. However, \cite{nadimpalli2023pauli} also demonstrated that their low-degree \QACZ~concentration results hold for circuits which do not perform clean computation and even for circuits with dirty ancilla qubits. Therefore, it could be interesting to explore the feasbility of our learning algorithm in these different ancilla settings. 

\paragraph{Proper Learning.} While our unitary synthesis procedure guarantees a poly-logarithmic depth circuit implementation, we believe that proper learning of the \QACZ~circuit should be possible. Our current circuit synthesis approach is via brute-force search and an $\varepsilon$-net, which we believe leaves room for improvement:
\begin{openprob}[Quasi-polynomial proper learning \QACZ]
    Given the learned unitary $C_\text{sew}$, which is $1/\poly(n)$-close to \QACZ~circuit $C$, does there exist a quasi-polynomial time algorithm to learn a \QACZ~circuit implementing unitary $C_\text{sew}^*$ such that
    $\davg(C_\text{sew}^*, C \otimes C^\dagger) \leq 1/\poly(n)$?
\end{openprob}
Although it is widely believed that efficient \ACZ~proper learners do not exist, \ACZ~and \QACZ~with sublinear ancilla are fundamentally incomparable\footnote{To see why: \cite{anshu2024computational} showed that Parity$\notin$\QACZ[sublinear]. Since Parity$=$Fanout quantumly, so Fanout $\notin$ \QACZ[sublinear]. However, Fanout$\in$\ACZ~by definition.}.
Thus, efficient proper learning of \QACZ~with sublinear ancilla would not contradict the classical belief.

\paragraph{Time Complexity.}
Classically, \cite{kharitonov1993cryptographic} showed that \cite{linial1993constant}'s quasi-polynomial complexity for learning \ACZ~is optimal, under the standard classical cryptographic hardness-of-factoring assumption. Since factoring is not quantumly hard, \cite{arunachalam2021quantum} established a quasi-polynomial time lower-bound for \emph{quantum} learning of \ACZ~via a reduction to Ring Learning with Errors (\RLWE). Although there are fundamental differences between \ACZ~and \QACZ, we believe that it should be possible (potentially through similar crypotgraphic reductions) to prove a quasi-polynomial lower-bound for learning \QACZ. This would imply optimality of our algorithm.

\paragraph{Sample Complexity.} Our learning algorithm's core subroutine can be interpreted as learning a junta approximation of \QACZ~Heisenberg-evolved Pauli observables. To this end, there is a growing literature of sophisticated algorithms for quantum junta learning \citep{chen2023testing, bao2023nearly, gutierrez2024}. 
While these results do not directly apply to learning the entire $n$-qubit unitary of a \QACZ~circuit, they open up intriguing possibilities for improving our algorithm's sample complexity with respect to parameters such as circuit size.

\paragraph{Beyond \QACZ.} Finally, our work contributes to the burgeoning field of research seeking to identify provably trainable and learnable quantum circuit families.
Our results extend \cite{huang2024learning}'s efficient learnability of \QNCZ~to the broader class of \QACZ. This advance raises intriguing questions for future exploration.
For instance, since \QACZ~requires linear depth to implement in 1D geometry, can we efficiently learn other polynomial-depth circuit families in the 1D geometry? Furthermore, since \QACZ $\subseteq$ \QNCO $\subseteq \mathsf{QNC}$, a natural progression would be to investigate efficient learnability of polylog-depth $\mathsf{QNC}$ circuits  or other superclasses of \QACZ. By uncovering rigorous algorithms for training and learning quantum circuits, we may enable automated design of quantum circuits, protocols, and algorithms, in addition to characterization of experimental quantum devices.

\section{Acknowledgements}
The authors thank Zeph Landau, Fermi Ma, Jarrod McClean, and Ewin Tang for helpful discussions. The authors would also like to thank an anonymous reviewer for suggesting we add a hardness of learning \QACZ~argument. FV is supported by the Paul and Daisy Soros Fellowship for New Americans as well as the National Science Foundation Graduate Research Fellowship under Grant No. DGE 2146752.

\bibliographystyle{alpha}
\bibliography{biblio}

\appendix

\newpage

\section{Preliminaries} \label{sec:prelims}

In this section, we will briefly establish notation and important concepts to be used throughout the paper.  This manuscript will assume familiarity with the basics of quantum computational and information theory. For more background, we refer the interested reader to \cite{nielsen2010quantum, wilde2013quantum}. In particular, the preliminary sections of \cite{nadimpalli2023pauli,huang2024learning} are especially relevant to this work. For more information on tensor networks, as used in \Cref{fig:tens_choi_heis}, we refer the reader to \cite{bridgeman2017hand}.

In terms of notation, we will generally use $n$ to refer to the number of qubits in the computational register of a quantum system and $a$ to refer to the number of ancilla qubits.  Let $[n]=\{1,2,\dots, n \}$ denote the set of qubit indices. We will use the notation $I_{k}$ to refer to an identity matrix of dimension $2^k \times 2^k$ and, when the dimension is implied by the context, may drop the $k$ subscript altogether.

Given an $n$-qubit unitary $U$ corresponding to a quantum circuit, we will define the circuit's corresponding unitary channel as 
\begin{align}
    \calU(\rho) = U \rho U^\dag.
\end{align}
Implicitly, we will assume that unitary channels are perfectly implemented, without any noise. Furthermore, in this work, ancilla qubits are assumed to be restricted to input $\ketbra{0}{0}$. For the same quantum circuit, implementing unitary $U$, and a given Hermitian measurement observable $M$, the circuit's Heisenberg-evolved observable (corresponding to the dual channel) is defined as
\begin{align}
    O_M = U^\dag M U.
\end{align}
Note that if the measurement observable $M$ is unitary, then $O_M$ is also unitary.

\subsection{Circuit Classes} 
This work will refer to several classical and quantum circuit classes. We briefly review them here for ease of reference. 

Classically, \NC$^k$ is the class of circuits of depth $\calO(\log^k n)$ and size $\poly (n)$, comprised of {\tt AND}, {\tt OR}, and {\tt NOT} gates of fan-in $\leq 2$. The circuit class \NC$=\bigcup_{k\geq 1} \mathsf{NC}^k$ refers to the union of all these circuit classes. The circuit classes \AC$^k$~and \AC~are defined analogously, but allow for {\tt AND} and {\tt OR} gates of unbounded fan-in.

Quantumly, \QNC$^k$ is the class of quantum circuits of depth $\calO(\log^k n)$, comprised of unlimited arbitrary single-qubit gates and polynomially many $CZ$ gates (that act on 2-qubits). The circuit class \QNC$=\bigcup_{k\geq 1} \mathsf{QNC}^k$ refers to the union of all these circuit classes. The circuit classes \QAC$^k$~and \QAC~are defined analogously, but allow for $CZ$ gates that simultaneously operate on an unbounded number number of qubits.

\subsection{The Average-Case Distance Measure} \label{sec:dist_meas}

This work will leverage average-case distance measures, so we will now review their important properties. The presentation and intuition are largely based on prior works \cite{huang2024learning, Nielsen_2002}. To begin, a standard notion of distance for quantum states is the fidelity measure. 
\begin{definition}[Fidelity] 
    Given two quantum states $\rho, \sigma$ the fidelity of the states is defined as
    \begin{align}
        \calF(\rho, \sigma) = \Tr(\sqrt{\sqrt{\rho}\sigma \sqrt{\rho}})^2.
    \end{align}
    If one of the states is a pure state, e.g. $\sigma=\ketbra{\psi}{\psi}$, then the fidelity expression reduces to 
    \begin{align}
        \calF(\rho, \sigma) = \bra{\psi} \rho \ket{\psi}.
    \end{align}
\end{definition}
\noindent Note that for any input $(\rho, \sigma)$, the fidelity is bounded as $\calF(\rho, \sigma)\in[0,1]$. Furthermore, the maximal value of 1 is obtained if and only if the states are identical, i.e. $\calF(\rho,\rho)=1$.

For quantum channels, we consider the standard Haar distance measure, which is the same notion of average-case distance for comparing quantum channels as utilized in \cite[Definition 3]{huang2024learning}.
\begin{definition}[Average-Case Distance] \label{def:avg_case_dist}
The average-case distance between two $n$-qubit CPTP maps $\calE_1$ and $\calE_2$ is defined as
\begin{align}
    \davg(\calE_1, \calE_2)=\underset{\ket{\psi}\sim \textnormal{Haar}}{\ex}\left[1-\calF\left(\calE_1(\ketbra{\psi}{\psi}),\calE_2(\ketbra{\psi}{\psi})\right)\right],
\end{align}
where $\ket{\psi}$ is sampled from the Haar (uniform) measure and  $\calF$ is the fidelity. 
\end{definition}
\noindent Intuitively, this notion of average-case distance measures how distinct the two channel outputs are, for pure input states averaged over the Haar (uniform) measure. In the case of unitary channels the Haar distance measure simplifies to the average gate fidelity measure. As such, we will often abuse notation for unitary channels and write $\davg(U_1, U_2)$ to mean $\davg(\calU_1, \calU_2)$.
\begin{fact}[Average Gate Fidelity - \cite{Nielsen_2002}] \label{fact:uni_chans_dist}
    For unitaries $U_1$ and $U_2$, with corresponding unitary channels $\calU_1$ and $\calU_2$, the average-case distance satisfies
    \begin{align}
        \davg(\calU_1, \calU_2) = \frac{2^n}{2^n+1} \left(1-\frac{1}{4^n}\left|\Tr(U_1^\dag U_2)\right|^2\right).
    \end{align}
\end{fact} 


In this work we will also leverage properties of the well-established normalized Frobenius and global phase-invariant distance measures for unitaries.
\begin{definition}[Normalized Frobenius Distance] 
The normalized Frobenius distance between two $n$-qubit unitaries $U$ and $V$ is defined as
\begin{align}
    \frobd(U, V)= \frac{1}{2^n}\|U-V\|_F^2
\end{align}
\end{definition}
\begin{definition}[Global Phase-Invariant Distance]
The global phase-invariant distance between two $n$-qubit unitaries $U$ and $V$ is defined as
\begin{align}
    \gpid(U, V)=\min_{\phi \in \mathbb{R}}~\frac{1}{2^n}\|e^{i\phi}U-V\|_F^2
\end{align}
\end{definition}

\noindent As distance measures, both satisfy the triangle inequality. Furthermore, from these definitions, it trivially follows that the normalized Frobenius distance is lower-bounded by the global phase-invariant distance.
\begin{fact} \label{fact:gpid_ub}
    For unitaries $U_1$ and $U_2$, the global phase-invariant distance is upper-bounded by the normalized Frobenius distance,
    \begin{align}
        \gpid(U_1,U_2) \leq  \frobd(U_1,U_2).
    \end{align}
\end{fact}
\noindent \cite{huang2024learning} also showed that this average-case distance can be upper-bounded by the Frobenius norm distance between the unitaries.
\begin{fact} \label{fact:davg_ub}
    For unitaries $U_1$ and $U_2$, with corresponding unitary channels $\calU_1$ and $\calU_2$, the average-case distance is upper-bounded by the global phase-invariant distance,
    \begin{align}
        \davg(\calU_1,\calU_2) \leq  \gpid(U_1,U_2).
    \end{align}
\end{fact}

\subsection{Pauli Analysis} 

In this work, we will make frequent use of concepts from so-called ``Pauli analysis'' or ``quantum Boolean functions'', as originally defined by \cite{montanaro2010quantum}, which is a quantum analog of classical analysis of Boolean functions \cite{o2014analysis}.

Central to Pauli analysis are the four standard Pauli operators:
\begin{align}
    I := \begin{pmatrix}
        1 & 0 \\
        0 & 1
    \end{pmatrix}, \quad X := \begin{pmatrix}
        0 & 1 \\
        1 & 0
    \end{pmatrix}, \quad Y := \begin{pmatrix}
        0 & -i \\
        i & 0
    \end{pmatrix}, \quad Z := \begin{pmatrix}
        1 & 0 \\
        0 & -1
    \end{pmatrix}
\end{align}
Note that the identity matrix is often thought of as the ``trivial'' Pauli operator, as it has no effect on a quantum system. We will denote the set of single-qubit Pauli operators as $\pauli = \{I,X,Y,Z\}$. Let $\pauli^n = \{I,X,Y,Z\}^{\otimes n}$ denote the set of all $n$-qubit Pauli strings. 

In general, given an $n$-qubit Pauli string $Q\in \pauli^n$, we use $Q_i$ to refer to the $i^\text{th}$ Pauli of Pauli string $Q$. Similarly, given a set of indices $\calS \subseteq [n]$, $Q_\calS$ refers to the Pauli sub-string
\begin{align}
    Q_\calS = \bigotimes_{i \in [n]} Q_i^{\delta\{i\in \calS\}},
\end{align}
with the convention that $Q_i^0 = I$. Note, however, that we will abuse notation and, in the case of single-qubit Pauli $P \in \pauli$ use $P_i$ to refer to the ``single-qubit Pauli''
\begin{align}
    P_i := I^{\otimes (i-1)} \otimes P \otimes I^{\otimes (n-i)},
\end{align}
which applies $P$ to the $i^\text{th}$ qubit in an $n$-qubit system. We define the support of a Pauli string $Q \in \pauli^n$ to be the set of qubit indices upon which the Pauli acts non-trivially, that is
\begin{align}
    \supp(Q) = \{i \in [n]: Q_i \neq I\}.
\end{align}
The degree of a Pauli string $Q$ is the size of its support, 
\begin{align}
    |Q|=|\supp(Q)|.
\end{align}

Central to Pauli analysis is the observation that the set of Pauli strings $\pauli^n$ forms an orthonormal basis for any $n$-qubit quantum unitary $U \in U(2^n)$,
\begin{align} \label{eqn:unitary_decomp}
    U = \sum_{P \in \pauli^n} \widehat{U}(P) \cdot P, \quad \text{ where } \widehat{U}(P)=\frac{1}{2^n} \Tr (U^\dagger P).
\end{align}
Note that $\widehat{U}(P)$ are typically referred to as the Pauli coefficients of $U$. 

For a subset of Pauli strings, $\calS \subseteq \pauli^n$, we define the unitary's Pauli weight on that subset as
\begin{align}
    \wt^{\in \calS} [U] = \sum_{Q \in \calS} |\widehat{U}(Q)|^2.
\end{align} 
For notational convenience, we will use the notation $\wt^{\notin \calS}$ to refer to the weight in subset $\bar{\calS} = \calP^n \backslash \calS$ and the notation $\wt^{>k}$ to refer to the weight of Paulis with degree greater than $k$, i.e. $\calS = \{P \in \pauli^n: |P| > k\}$. Finally, note that Parseval's formula holds for the Pauli decomposition.
\begin{fact}[Parseval's Formula] \label{fact:parseval}
    For unitary $U \in U(2^n)$, 
    \begin{align}
        \frac{1}{2^n} \| U \|_F^2 = \sum_{P \in \pauli^n} |\widehat{U}(P)|^2 = \wt^{\in \pauli^n}[U] = 1
    \end{align}
\end{fact}

Finally, important to this work, we will distinguish between the support and degree of a unitary. The support of a unitary is the set of all qubits upon which the unitary acts non-trivially. Formally, for the $n$-qubit unitary $U$ defined in \Cref{eqn:unitary_decomp}, its support is the union of the supports of all the Pauli strings $P$ for which $\widehat{U}(P) \neq 0$, i.e.
\begin{align}
    \supp(U) = \bigcup_{P\in\paulis: \hat{U}(P) \neq 0} \{i \in [n]: P_i \neq I \}.
\end{align}
If $n$-qubit unitary has support $k$, where $k$ is small relative to $n$, we refer to $U$ as low-support. In analogy to the classical setting, this is similar to a $k$-junta. Meanwhile, the degree of a unitary is equal to the maximum degree of any Pauli string $P$ for which $\widehat{U}(P) \neq 0$. Formally, for the $n$-qubit unitary $U$ defined in \Cref{eqn:unitary_decomp}, its degree is
\begin{align}
    |U| = \max_{P\in\paulis: \hat{U}(P) \neq 0} |P|.
\end{align}
If $n$-qubit unitary has degree $k$, where $k$ is small relative to $n$, we refer to $U$ as low-degree. It is clear that any unitary with support $k$ has degree $k$. However, a unitary with degree $k$ could have, for example, non-zero Pauli coefficients on all Pauli strings of degree $k$ and, thus, support $n$. Therefore, a unitary of degree $k$ does not necessarily have support $k$. This implies that any low-support unitary is neccesarily a low-degree unitary, but any low-degree unitary is not neccesarily a low-support unitary. Note that this is analogous to how, in the classical setting, all $k$-juntas necessarily have degree-$k$, but a degree-$k$ Boolean function is not necessarily a $k$-junta.

\subsection{Classical Shadow Tomography} 
Central to our learning results will be the classical shadow tomography procedure of \cite{huang2020predicting}, which enables efficient learning of $\Tr(O_i \rho)$ for an arbitrary quantum state $\rho$ and a set of observables $\{O_i\}_i$. In particular, we will leverage the improved sample-complexity achieved by the shadow-norm result of \cite{huang2023learning} for Pauli observables.

\begin{lemma}[Classical Shadows for Low-Degree Pauli Observables - \cite{huang2020predicting,huang2023learning}] \label{thm:shadow_tomog}
    Assume we are given error parameter $\epsilon >0$, failure probability $\delta \in [0,1]$, degree $\ell \geq 1$, and
    \begin{align}
        N = \calO \left( \frac{3^\support}{\epsilon^2} \log \left(\frac{n^\support}{\delta}\right)\right)
    \end{align}
    copies of unknown $n$-qubit state $\rho$ that we can make random Pauli measurements on. Let $\calM$ be a set of $n$-qubit Pauli matrices of degree $\leq \ell$, i.e.
    \begin{align}
        \calM \subseteq \{P \in \pauli^n: |P| \leq \ell\}.
    \end{align}
    With probability at least $1-\delta$, we can output an estimate $\widetilde{s}(P)$ for each $P \in \calM$ such that 
    \begin{align}
        |\widetilde{s}(P)-\Tr(P\rho)| \leq \epsilon.
    \end{align}
    The computational complexity of this procedure scales as $\calO(|\calM| \cdot N)$.
\end{lemma}

\subsection{Unitary Projection}
Our sewing procedure will require an efficient procedure for projecting non-unitary matrices onto unitary matrices. Specifically, we require a projection which minimizes the Frobenius norm distance. As such, for an arbitrary matrix $A$, we will define $\proj_U (A)$ to be the projection of $A$ onto the unitary matrix that minimizes Frobenius norm distance, 
\begin{align} \label{eqn:proj_u}
    \proj_U (A) = \underset{B\in U(2^n)}{\min} \|A - B\|^2_F.
\end{align}
Note that this minimization task is the well-established ``Orthogonal Procrustes Problem" and has an efficiently computable, simple solution.
\begin{fact}[Orthogonal Procrustes Problem] \label{fact:project}
    For $m\times m$ matrix $A$ with singular value decomposition $A=U \Sigma V^\dagger$, 
    \begin{align}
        \proj_U (A) = U V^\dagger,
    \end{align}
    which can be computed in $\calO(m^3)$ time.
\end{fact}

\subsection{The $CZ_k$ Gate}
In this work, we will use $CZ_k$ to refer to a control-$Z$ gate acting on $k$-qubits. Note that, unlike other quantum ``controlled'' operations, the $CZ_k$ does not need to distinguish between target and control qubits, since it only applies a phase when activated by the all ones state. Mathematically, the $CZ_k$ gate can be described by its action on the computational basis states $\ket{i}\in\{\ket{0},\ket{1}\}^{\otimes k}$, as
\begin{align}
    CZ_k \ket{i} = \begin{cases}
        - \ket{i}, &\textnormal{ if } \ket{i} = \ket{1_k} \\
        \ket{i}, &\textnormal{ otherwise}
    \end{cases},
\end{align}
where we use $\ket{1_k}=\ket{111...1}$ to denote the all ones state. Thus, the  $CZ_k$ unitary can be decomposed in terms of projectors as
\begin{align} \label{eqn:projector}
    CZ_k = -\ketbra{1_k}{1_k}+\sum_{\substack{j\in\{0,1\}^k\\j\neq 1_k}} \ketbra{j}{j}.
\end{align}
Alternatively, the $CZ_k$ unitary can be expressed in terms of its Pauli decomposition.
\begin{lemma}[$CZ_k$ Pauli Decomposition] \label{thm:czk_decomp}
The Pauli decomposition of the $CZ_k$ gate acting on $k$ qubits is given by 
\begin{align}
    CZ_k = \sum_{P_Z\in\{I, Z\}^{\otimes k}} \widehat{\alpha}_{CZ_k}(P_Z) \cdot P_Z,
\end{align}
with Pauli coefficients
\begin{align}
    \widehat{\alpha}_{CZ_k}(P_Z) = 
    \begin{cases}
        1-2^{-k+1}, & \textnormal{ if } \deg(P_Z)=0 \\
        2^{-k+1}, & \textnormal{ if } \deg(P_Z)=\textnormal{odd}, \\
        -2^{-k+1}, & \textnormal{ if } \deg(P_Z)\neq 0 \textnormal{ and } \deg(P_Z)=\textnormal{even}. 
    \end{cases}
\end{align}
Note that we use notation $P_Z$ to emphasize that Pauli strings consisting of any $X$ or $Y$ Paulis have zero Fourier mass.
\end{lemma}
\begin{proof}[Proof of \Cref{thm:czk_decomp}]
    Using the $CZ_k$ projector decomposition given in \Cref{eqn:projector}, we can calculate the $CZ_k$ Pauli coefficient of for any Pauli string $P \in \{I,Z\}^{\otimes k}$ as
    \begin{align}
        \widehat{\alpha}_{CZ_k}(P) &= \frac{1}{2^k} \Tr \left(P^\dag CZ_k\right)\\
        &= -\frac{1}{2^k} \Tr \left(P \ketbra{1_k}{1_k}\right)+\frac{1}{2^k}\sum_{\substack{j\in\{0,1\}^k\\j\neq 1_k}}\Tr \left( P\ketbra{j}{j}\right) \\
        &= -\frac{1}{2^k} \bra{1_k}P \ket{1_k}+\frac{1}{2^k}\sum_{\substack{j\in\{0,1\}^k\\j\neq 1_k}}\bra{j}P\ket{j} \\
        &=-\frac{1}{2^k} \bra{1_k}P \ket{1_k}+\frac{1}{2^k}\left(\Tr(P)-\bra{1_k}P \ket{1_k}\right) \\
        &=\frac{1}{2^k}\Tr(P)-\frac{2}{2^k}\bra{1_k}P \ket{1_k} \\
        &=\frac{1}{2^k}\cdot 2^k \cdot \delta(P=I_k)-\frac{1}{2^{k-1}}\cdot (-1)^{\delta(\textnormal{deg}(P)=\textnormal{odd})} \\
        &=\delta(P=I_k)+\frac{1}{2^{k-1}}\cdot (-1)^{\delta(\textnormal{deg}(P)=\textnormal{even})}.
    \end{align}
\end{proof}
\noindent Note that this $k$-qubit decomposition can straightforwardly be extended to a $CZ_k$ gate acting on a $k$-qubit subset of an $n$-qubit system, as
\begin{align}
    CZ_k\otimes I_{n-k}  = \sum_{P_Z\in\{I, Z\}^{\otimes k}} \widehat{\alpha}_{CZ_k}(P_Z) \cdot P_Z\otimes I_{n-k},
\end{align}
where it is assumed without loss of generality that the $CZ_k$ gate acts on the first $k$ qubits of the system. From hereon, when we write $CZ_k$, the identity on the remaining $n-k$ qubits, i.e. $I_{n-k}$, will be assumed.

\section{Choi Representations and Heisenberg-Evolved Observables} \label{sec:choi_obs_relate}
The work of \cite{nadimpalli2023pauli} studied the Choi representation of quantum channels, while the work of \cite{huang2024learning} considered the Heisenberg-evolved observables of quantum circuits. Similar to \cite{huang2024learning}, in this work, we consider \QACZ~Heisenberg-evolved single-qubit Pauli observables. However, we will now establish that they are in fact closely related to the single-output \QACZ~channel Choi representations of \cite{nadimpalli2023pauli}.

In particular, the work of \cite{nadimpalli2023pauli} studied channels of the form
\begin{align}
    \calE_C(\rho) = \Tr_{n-1}(C\rho C^\dag),
\end{align}
where $C$ is the unitary implemented by a \QACZ~circuit and $\rho$ is a density matrix. Their results were presented in terms of the channel's Choi representation, given by
\begin{align}
    \Phi_{\calE_C} = \sum_{x,y \in \{0,1\}^n} \ketbra{x}{y} \otimes \calE_C \left(\ketbra{x}{y}\right). 
\end{align}
Note that $\Phi_{\calE_C}$ operates on $n+1$-qubits, where we denote the $n$-qubit register corresponding to the channel input as the ``in" register and the single qubit corresponding to the channel output as the ``out" register. From the definition of the Choi representation, it follows that
\begin{align} \label{eqn:choi_rep_chan}
    \calE_C(\rho)=\Tr_\text{in}\left(\Phi_{\calE_C} \left(I_\text{out} \otimes \rho^\top_\text{in}\right)\right).
\end{align}

Measuring the output of this channel with respect to a single-qubit observable $O$ results in the expectation value $\Tr(O \calE_C(\rho))$. Via algebraic manipulation of this expectation, we can solve for the dual channel $\calE^\dag_C (O)$, 
\begin{align}
    \Tr(O \calE_C(\rho)) &= \Tr\left(O \cdot \Tr_{n-1}(C\rho C^\dag)\right) = \Tr\left((I_{n-1}\otimes O) C\rho C^\dag\right)=\Tr(\underbrace{(C^\dag(I_{n-1}\otimes O) C)}_{\calE_C^\dag(O)}\rho ).
\end{align}
Therefore, the Heisenberg-evolved single-qubit observable of this single-output \QACZ~circuit is 
\begin{align} \label{eqn:heis_obs_def}
    \calE_C^\dag (O)= C^\dag(I_{n-1}\otimes O) C.
\end{align}
Alternatively, leveraging \Cref{eqn:choi_rep_chan},
\begin{align}
    \Tr\left(O\calE_C(\rho)\right) &= \Tr \left(O \Tr_\text{in}\left(\Phi_{\calE_C} (I_\text{out} \otimes \rho^\top_\text{in})\right)\right) \\
    &= \Tr \left((O_\text{out}\otimes I_\text{in}) \Phi_{\calE_C} (I_\text{out} \otimes \rho^\top_\text{in})\right) \\
    &= \Tr \left(\Tr_\text{out}\left((O_\text{out}\otimes I_\text{in}) \Phi_{\calE_C}\right)  \rho^\top_\text{in}\right) \\
    &= \Tr \bigg(\underbrace{\Tr_\text{out}\left((O_\text{out}\otimes I_\text{in}) \Phi_{\calE_C}\right)^\top}_{\calE_C^\dag(O)}  \rho_\text{in}\bigg),
\end{align}
we can also express the Heisenberg-evolved observable in terms of the Choi representation, as
\begin{align} \label{eqn:heis_choi}
   \calE_C^\dag (O)= \Tr_\text{out}\left((O_\text{out}\otimes I_\text{in}) \Phi_{\calE_C}\right)^\top.
\end{align}
Thus, \Cref{eqn:heis_obs_def} and \Cref{eqn:heis_choi} directly establish the link between \QACZ~Heisenberg-evolved single-qubit observables and single-qubit output \QACZ~Choi representations. 

Intuitively, the Choi representation contains a full description of the channel and can be used to calculate the channel expectation for any input state and observable pair $(\rho, O)$. Meanwhile, the Heisenberg-evolved observable restricts the channel output to a single measurement observable $O$, but can be used to compute the expectation for any input state $\rho$. This relationship is illustrated by the tensor network diagrams in \Cref{fig:tens_choi_heis}.

\begin{figure}[t]
    \centering
    \includegraphics[width=0.75\linewidth]{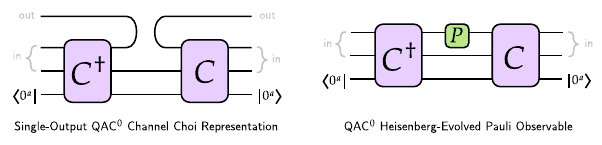}
    \caption{Tensor network diagrams illustrating the similar structures of the single-output \QACZ~channel Choi representations studied in the work of \cite{nadimpalli2023pauli} and the \QACZ~Heisenberg-evolved Pauli observables considered in \cite{huang2024learning} and in this work. Here, $C$ is the \QACZ~circuit's unitary and $P$ is a single-qubit observable.} 
    \label{fig:tens_choi_heis}
\end{figure}

Finally, we can explicitly relate the Pauli coefficients of each \QACZ~Heisenberg-evolved single-qubit Pauli observable to a Fourier coefficient of the single-output \QACZ~channel Choi representation. 
\begin{proposition} \label{thm:relate_coeffs}
    For an $n$-qubit Pauli $Q\in\pauli^n$, the Pauli coefficient of the Heisenberg-evolved observable $O_{P_\text{out}}$ is related to that of the single-output Choi representation $\Phi_{\calE_C}$ as
    \begin{align}
        \widehat{O}_{P_\text{out}} (Q) = 2 \cdot \widehat{\Phi}_{\calE_C}(P_\text{out} \otimes Q) \cdot (-1)^{\delta \{Q\text{ has an odd \# of Pauli } Ys\}}.
    \end{align}
\end{proposition}
\begin{proof}[Proof of \Cref{thm:relate_coeffs}]
     For an $n$-qubit Pauli $Q\in\pauli^n$, leveraging \Cref{eqn:heis_choi}, 
    \begin{align}
        \widehat{O}_{P_\text{out}} (Q) &= \frac{1}{2^n} \Tr(O_{P_\text{out}} \cdot Q) \\
        &= \frac{1}{2^n} \Tr\left(\Tr_\text{out}((P_\text{out} \otimes I_\text{in}) \Phi_{\calE_C} )^\top Q\right) \\
        &= \frac{1}{2^n} \Tr\left(\Tr_\text{out}\left((P_\text{out} \otimes I_\text{in}) \Phi_{\calE_C} (I_\text{out} \otimes Q^\top)\right)\right) \\
         &= \frac{1}{2^n} \Tr\left(\Tr_\text{out}\left((P_\text{out} \otimes Q^\top) \cdot \Phi_{\calE_C} \right)\right) \\
         &= 2 \cdot \frac{1}{2^{n+1}} \Tr\left((P_\text{out} \otimes Q^\top) \cdot \Phi_{\calE_C} \right)\\
         &= 2 \cdot \widehat{\Phi}_{\calE_C}\left(P_\text{out} \otimes Q^\top\right) \\
         &= 2 \cdot \widehat{\Phi}_{\calE_C}\left(P_\text{out} \otimes Q\right) \cdot (-1)^{\delta \{Q\text{ has an odd \# of Pauli } Ys\}}.
    \end{align}
    This completes the proof.
\end{proof}

\section{Concentration of \texorpdfstring{\QACZ}~ Heisenberg-Evolved Observables} \label{sec:heis_evolve_obs}

Assume we are given an $n$-qubit \QACZ~circuit governed by unitary, 
\begin{align} \label{eqn:qac0_circ}
    C = U_1 \cdot CZ_{k_1} \cdot U_2  \cdot CZ_{k_2} \cdots  CZ_{k_{m-1}} \cdot U_{m} \cdot CZ_{k_m} \cdot U_{m+1} = \prod_{i=1}^m \left(U_{i} \cdot CZ_{k_i}\right)\cdot U_{m+1}.
\end{align}
where $\{k_1, ..., k_m\}\in [\gaterem,n]$ are the sizes of the large $CZ_{k_i}$ gates in the circuit, such that $\gaterem\geq 2$ is the size of the smallest. For reasons that will become apparent later, we will assume that $\gaterem= \calO (\log(n))$. The unitaries $U_i$ correspond to circuits consisting of arbitrary single-qubit gates and $CZ_k$ gates of width $k \leq \gaterem$.
The single qubit Heisenberg-evolved observables of circuit $C$ are denoted 
\begin{align}
    O_{P_i} = C^\dagger P_i C,
\end{align}
where $P_i\in\{X,Y,Z\}_i \otimes I_{[n]\backslash i}$ is a non-trivial Pauli on the $i$-th qubit. Since the \QACZ~circuit can have $CZ_k$ gates of unbounded width, in the worst case it is supported on all qubits, i.e.
\begin{align}
    |\supp(O_{P_i})| \leq \calO(n).
\end{align} 
Finally, we denote the Pauli decomposition of this observable as
\begin{align}
    O_{P_i} = \sum_{Q\in\calP^n} \widehat{O}_{P_i} (Q) \cdot Q.
\end{align}

In this section, we will show that circuit $C$ is in fact well approximated by the same circuit with all the $CZ_{k_i}$ gates of size $k_i \geq \gaterem= \calO (\log(n))$ removed. We will denote this approximate circuit as 
\begin{align}
    \widetilde{C} = U_1 \cdot U_2  \cdots   \cdot U_{m} \cdot U_{m+1} = \prod_{i=1}^{m+1} U_i
\end{align}
and its single-qubit Heisenberg-evolved observables as 
\begin{align}
    O^*_{P_i} = \widetilde{C}^\dagger P_i \widetilde{C}.
\end{align}
Since the circuit only has gates of width $\leq \gaterem = \calO (\log(n))$,  note that it has a much smaller support than $C$, i.e. 
\begin{align}
    \support=|\supp(O^*_{P_i})|\leq\calO\left(\log^dn\right),
\end{align}
where $d$ is the circuit depth (which is constant for \QACZ~circuits). The Pauli decomposition of these approximate observables will be expressed as
\begin{align}
    O^*_{P_i} = \sum_{Q\in\calP^n:|Q|<\support} \widehat{O}^*_{P_i} (Q) \cdot Q.
\end{align}



With these definitions, we will now leverage the proof techniques of \cite{nadimpalli2023pauli} to establish low-degree and low-support concentration of Heisenberg-evolved observables, so as to obtain an upperbound for $\wt^{>k}[O_{P_i}]$ by relating it to $O^*_{P_i}$.

\subsection{Low-Degree Concentration} \label{sec:ub_weight}
  
We begin by demonstrating that the low-degree spectral concentration result of \cite[Theorem 21]{nadimpalli2023pauli} for single-output \QACZ~channel Choi representations also holds for \QACZ~Heisenberg-evolved single-qubit Pauli observables. Note, that this proof will only be for \QACZ~circuits without ancilla qubits. For more discussion on the effect of ancillas we refer the reader to \Cref{sec:ancillas}.
\begin{proposition}[Low-Degree Concentration] \label{thm:low_deg_conc}
    Suppose $C$ is a depth-$d$, size-$s$ \QACZ~circuit acting on $n$ qubits. Let $O_{P_i} = C^\dagger P_i C$ be a Heisenberg-evolved single-qubit Pauli observable. Then for every degree $k \in [n]$,
    \begin{align}
        \wt^{>k}[O_{P_i}] \leq \calO \left(s^2 2^{-k^{1/d}}\right).
    \end{align}
\end{proposition}

At a high-level, this proof will follow that of \cite[Theorem 21]{nadimpalli2023pauli} and consist of two key steps. First, we will establish that if the \QACZ~circuit has no $CZ$ gates of width greater than $k^{1/d}$, then $\wt^{>k}[O_{P_i}]=0$. Second, we will show that removing these ``large'' $CZ$ gates does not significantly change the Heisenberg-evolved observable under the average-case measure. 

Analogous to \cite[Lemma 20]{nadimpalli2023pauli}, we begin by showing that the weight spectrum of the \QACZ~Heisenberg-evolved observable is zero for any degree greater than the size of the observable's support.
\begin{lemma} \label{thm:weight_zero}
    Let $O_{P_i}$ be an observable corresponding to a circuit $C$ measured with respect to $P_i$. If $O_{P_i}$ is supported on $\support$ qubits, i.e. $|\supp(O_{P_i})|=\support$, then the observable's weight is zero for any degree $>\support$, i.e. 
    \begin{align}
        \wt^{>\support}\left[O_{P_i}\right]=0.
    \end{align}
\end{lemma}
\begin{proof}[Proof of \Cref{thm:weight_zero}]
    Given the circuit $C$ and an observable $P_i$, we will decompose the circuit as
    \begin{align}
        C = D_{P_i} L_{P_i},
    \end{align}
    where $L_{P_i}$ is the unitary corresponding to all gates in the circuit that are in the backwards light-cone of $P_i$ and $D_{P_i}$ is the unitary corresponding to the gates which are not. Thus, by the standard light-cone argument,
    \begin{align}
        C^\dagger P_i C = D_{P_i}^\dagger L_{P_i}^\dagger P_i  L_{P_i} D_{P_i} = L_{P_i}^\dagger P_i  L_{P_i}.
    \end{align}
    Plugging this into our expression for the weight of observable $O_{P_i}$ at degree $>k$,
    \begin{align}
        \wt^{>k}[O_{P_i}] &= \sum_{|Q|>k} \widehat{O}_{P_i}(Q)^2 = \sum_{|Q|>k} \frac{1}{4^n} \Tr(O_{P_i} Q)^2 = \sum_{|Q|>k} \frac{1}{4^n} \Tr(C^\dagger P_i C Q)^2 \\
        &= \sum_{|Q|>k} \frac{1}{4^n} \Tr(L_{P_i}^\dagger P_i L_{P_i}  Q)^2,
    \end{align}
    we see that the weight only depends on gates in the backwards light-cone of $P_i$. If $Q$ has degree greater than $|\supp(O_{P_i})|=\support$, this implies that there must exist at least one non-identity term in $Q$ which corresponds to an identity term in $O_{P_i}$, meaning $\Tr(L_{P_i}^\dagger P_i L_{P_i}  Q)=0$. Therefore, for $|Q|>\support$,  $\widehat{O}_{P_i}(Q)^2 = 0$, which implies that $\wt^{>\support}[O_{P_i}]=0$.
\end{proof}
\noindent Analogous to \cite[Lemma 23]{nadimpalli2023pauli}, we can also upper-bound the distance between the true observable, $O_{P_i}$, and the observable with ``large'' $CZ$ gates removed, $O_{P_i}^*$.
\begin{lemma} \label{thm:avg_dist}
    Let $O_{P_i}$ be an observable corresponding to a \QACZ~circuit $C$ measured with respect to $P_i$. Let $O^*_{P_i}$ be an observable corresponding to the \QACZ~circuit $\widetilde{C}$ (circuit $C$ with all $m$ $CZ_k$'s of size $k> \gaterem$ removed), measured with respect to $P_i$. The average-case distance between these two is observables is upper-bounded by
    \begin{align} \label{eqn:trunc_err}
        \frobd(O_{P_i},O_{P_i}^*) \leq \epsilon^* = \frac{9m^2}{2^{\gaterem}}.
    \end{align}
\end{lemma}
\begin{proof}[Proof of \Cref{thm:avg_dist}]
    Let us define the sub-circuits $C_j$ and $\widetilde{C}_j$ of $C$ and $\widetilde{C}$, respectively, as
    \begin{align}
        C_j &= U_j \prod_{i=j}^{m} \left(CZ_{k_i} U_{i+1}\right), \\
        \widetilde{C}_j &= \prod_{i=j}^{m} U_{i}.
    \end{align}
    Via unitary-invariance of the Frobenius norm and a hybrid argument,
    \begin{align}
        \left\| O_{P_i}-O_{P_i}^*\right\|_F &=  \left\| C_1^\dagger P_i C_1- \widetilde{C}_1^\dagger P_i \widetilde{C}_1 \right\|_F \\
        &= \left\| CZ_{k_1}^\dagger C_2^\dagger P_i C_2 CZ_{k_1} - \widetilde{C}_2^\dagger P_i \widetilde{C}_2 \right\|_F \\
        &= \left\| CZ_{k_1}^\dagger C_2^\dagger P_i C_2 CZ_{k_1}- C_2^\dagger P_i C_2 + C_2^\dagger P_i C_2 - \widetilde{C}_2^\dagger P_i \widetilde{C}_2 \right\|_F^2 \\
        &\leq \left\| CZ_{k_1}^\dagger C_2^\dagger P_i C_2 CZ_{k_1} - C_2^\dagger P_i C_2 \right\|_F^2 + \left\|C_2^\dagger P_i C_2 - \widetilde{C}_2^\dagger P_i \widetilde{C}_2 \right\|_F^2 \\
        & \quad \quad \quad \quad \vdots \\
        & \leq  \sum_{j=1}^m \left\| CZ_{k_j}^\dagger C_{j+1}^\dagger P_i C_{j+1} CZ_{k_j}- C_{j+1}^\dagger P_i C_{j+1} \right\|_F.  \label{eqn:sum_term}
    \end{align}
    Denoting $V_j = C_{j+1}^\dagger P_i C_{j+1}$ and using the fact that 
    \begin{align}
        CZ_{k_j} = I^{\otimes n} - 2\cdot I^{\otimes n-k_i}\otimes \ketbra{1^{k_j}}{1^{k_j}} = I - 2\ketbra{1^{k_j}}{1^{k_j}},
    \end{align}
    the terms in the summation of \Cref{eqn:sum_term} are upper-bounded as
    \begin{align}
        \left\| CZ_{k_j}^\dagger V_j CZ_{k_j}- V_j \right\|_F &= \sqrt{2\Tr(I)-2\Tr(V_j^\dagger CZ_{k_j} V_j CZ_{k_j}^\dagger)} \\
        &= \sqrt{2\Tr(I)-2\Tr(V_j^\dagger (I - 2\ketbra{1^{k_j}}{1^{k_j}}) V_j (I - 2\ketbra{1^{k_j}}{1^{k_j}}))} \\
        &= \sqrt{2\Tr(I)-2\Tr((V_j^\dagger - 2V_j^\dagger\ketbra{1^{k_j}}{1^{k_j}})  (V_j - 2V_j\ketbra{1^{k_j}}{1^{k_j}}))} \\
        &= \sqrt{2\Tr(I)-2\Tr(I - 4\ketbra{1^{k_j}}{1^{k_j}} + 4V_j^\dagger\ketbra{1^{k_j}}{1^{k_j}}V_j\ketbra{1^{k_j}}{1^{k_j}})} \\
        &= \sqrt{8\Tr(\ketbra{1^{k_j}}{1^{k_j}}) - 8\Tr(V_j^\dagger\ketbra{1^{k_j}}{1^{k_j}}V_j\ketbra{1^{k_j}}{1^{k_j}})} \\
        &= \sqrt{8\cdot 2^{n-k_j} - 8\Tr(V_j^\dagger\ketbra{1^{k_j}}{1^{k_j}}V_j\ketbra{1^{k_j}}{1^{k_j}})} \\
        &\leq  3\cdot 2^{(n-k_j)/2}.
    \end{align}
    Plugging this upper-bound back into \Cref{eqn:sum_term} and denoting $\gaterem=\min_j k_j$, we obtain the desired upper-bound for the Frobenius distance between the two observables, 
    \begin{align} \label{eqn:obs_ub}
        \frobd( O_{P_i},O_{P_i}^*) \leq \frac{1}{2^n} \left( \sum_{j=1}^m 3\cdot 2^{(n-k_j)/2} \right)^2 \leq \frac{1}{2^n} \left( 3m \cdot \max_j  2^{(n-k_j)/2} \right)^2   = \frac{9m^2}{2^\gaterem}.
    \end{align}
\end{proof}
\noindent Analogous to \cite[Claim 24]{nadimpalli2023pauli}, we will now show that the weight of $O_{P_i}$ can be upper-bounded by the weight of $O_{P_i}^*$ and its distance to $O_{P_i}^*$.
\begin{lemma} \label{thm:weight_ub}
    For any degree $k$, the weight of $O_{P_i}$ is upperbounded as
    \begin{align}
    \wt^{>k}[O_{P_i}]
    \leq \left( \wt^{>k}\left[O^*_{P_i}\right]^{1/2}+ \frac{1}{\sqrt{2^n}} \left\| O_{P_i}-O^*_{P_i}\right\|_F \right)^2.
\end{align}
\end{lemma}
\begin{proof}[Proof of \Cref{thm:weight_ub}]
    Note that the weight expression satisfies the triangle inequality.
    Thus, the weight of $O_{P_i}$ can be decomposed with respect to  $O^*_{P_i}$ as follows. 
    \begin{align}
        \wt^{>k}[O_{P_i}] &= \wt^{>k}\left[O^*_{P_i}+ (O_{P_i}-O^*_{P_i})\right] \\
        &\leq \left( \wt^{>k}\left[O^*_{P_i}\right]^{1/2}+ \wt^{>k}\left[O_{P_i}-O^*_{P_i}\right]^{1/2}\right)^2\\
        &\leq \left( \wt^{>k}\left[O^*_{P_i}\right]^{1/2}+ \frac{1}{\sqrt{2^n}} \left\| O_{P_i}-O^*_{P_i}\right\|_F \right)^2.
    \end{align}
    This concludes the proof.
\end{proof}

Leveraging \Cref{thm:weight_zero}, \Cref{thm:avg_dist}, and \Cref{thm:weight_ub}, we can now straightforwardly prove \Cref{thm:low_deg_conc}.
\begin{proof}[Proof of \Cref{thm:low_deg_conc}]
    Let $O_{P_i}^*$ be defined such that all $CZ$ gates of size $\geq \gaterem = k^{1/d}$ are removed. Since the \QACZ~circuit is depth-$d$, the support of $O_{P_i}^*$ is bounded as $|\supp(O_{P_i}^*)| < \left(k^{1/d}\right)^d = k$.
    Therefore, by \Cref{thm:weight_zero},
    \begin{align} \label{eqn:weight_bound_k}
        \wt^{>k}[O_{P_i}^*] = 0.
    \end{align}
    Furthermore, plugging $\gaterem = k^{1/d}$ into \Cref{thm:avg_dist}, the distance between $O_{P_i}$ and $O_{P_i}^*$ is bounded as
    \begin{align} \label{eqn:dist_bound_k}
        \frobd( O_{P_i}, O_{P_i}^*) \leq \frac{9m^2}{2^{k^{1/d}}} \leq \frac{9s^2}{2^{k^{1/d}}},
    \end{align}
    where we leveraged the fact that the total number of gates removed $m$, must be less than the size $s$ of the circuit. Plugging \Cref{eqn:weight_bound_k} and \Cref{eqn:dist_bound_k} into the degree-$k$ weight upper bound of \Cref{thm:weight_ub}, we obtain the desired result.
\end{proof}

\subsection{Low-Support Concentration}

For the computational efficiency of the learning algorithm to be presented in this work, it is crucial to establish that, beyond low-degree concentrated, the $O_{P_i}$ observables' weight spectrum is concentrated on a set of Paulis with small support, i.e. low-support concentrated. 

To begin, we modify \Cref{thm:weight_zero} to show that the weight of a Heisenberg-evolved observable is zero for all Paulis that lie outside of its support.

\begin{lemma} \label{thm:weight_zero_support}
    Let $O_{P_i}$ be an observable corresponding to a circuit $C$ measured with respect to $P_i$. Let $\mathcal{S}=\{P \in \paulis: P_i = I,~ \forall i \notin \supp(O_{P_i})\}$ denote the set of Pauli strings in $O_{P_i}$'s support. Then $O_{P_i}$'s weight is zero for all Paulis acting non-trivially outside of its support, i.e. 
    \begin{align}
        \wt^{\notin \mathcal{S}}\left[O_{P_i}\right]=0.
    \end{align}
\end{lemma}
\begin{proof}[Proof of \Cref{thm:weight_zero_support}]
    Given the circuit $C$ and an observable $P_i$, we will decompose the circuit as
    \begin{align}
        C = D_{P_i} L_{P_i},
    \end{align}
    where $L_{P_i}$ is the unitary corresponding to all gates in the circuit that are in the backwards light-cone of $P_i$ and $D_{P_i}$ is the unitary corresponding to the gates which are not. Thus, by the standard light-cone argument,
    \begin{align}
        C^\dagger P_i C = D_{P_i}^\dagger L_{P_i}^\dagger P_i  L_{P_i} D_{P_i} = L_{P_i}^\dagger P_i  L_{P_i}.
    \end{align}
    Plugging this into our expression for the weight of observable $O_{P_i}$ for Paulis outside the support,
    \begin{align}
        \wt^{\notin \mathcal{S}}[O_{P_i}] &= \sum_{Q\notin \mathcal{S}} \widehat{O}_{P_i}(Q)^2 = \sum_{Q\notin \mathcal{S}} \frac{1}{4^n} \Tr(O_{P_i} Q)^2 = \sum_{Q\notin \mathcal{S}} \frac{1}{4^n} \Tr(C^\dagger P_i C Q)^2 \\
        &= \sum_{Q\notin \mathcal{S}} \frac{1}{4^n} \Tr(L_{P_i}^\dagger P_i L_{P_i}  Q)^2,
    \end{align}
    we see that the weight only depends on gates in the backwards light-cone of $P_i$, and thus in the support of $O_{P_i}$. If $Q$ acts non-trivially outside $\mathcal{S}=\supp(O_{P_i})$, this implies that there must exist at least one non-identity term in $Q$ which corresponds to an identity term in $O_{P_i}$, meaning $\Tr(L_{P_i}^\dagger P_i L_{P_i}  Q)=0$. Therefore, for all $Q \notin \mathcal{S}$,  $\widehat{O}_{P_i}(Q)^2 = 0$, which implies that $\wt^{\notin \mathcal{S}}[O_{P_i}]=0$.
\end{proof}

\noindent Leveraging this result and a proof similar to that of \Cref{thm:weight_ub}, we achieve a low-support concentration result. In particular, we show that the weight of $O_{P_i}$ outside the set of Paulis in the support of $O_{P_i}^*$ is at most the distance between $O_{P_i}$ and $O_{P_i}^*$, which we proved in \Cref{thm:avg_dist} to decay with respect to the size of the support of $O_{P_i}^*$.

\begin{lemma} \label{thm:pauli_lb}
    For $\calS^* = \supp(O^*_{P_i})$, the weight of $O_{P_i}$ outside the support of $O_{P_i}^*$ is upper-bounded as
    \begin{align}
    \wt^{\notin\calS^*}[O_{P_i}]
    \leq \frobd( O_{P_i},O^*_{P_i}) \leq \epsilon^*.
\end{align}
\end{lemma}
\begin{proof}[Proof of \Cref{thm:weight_ub}]
    Note that the weight expression satisfies the triangle inequality. Thus, the weight of $O_{P_i}$ can be decomposed with respect to  $O^*_{P_i}$ as follows. 
    \begin{align}
        \wt^{\notin\calS^*}[O_{P_i}] &= \wt^{\notin\calS^*}\left[O^*_{P_i}+ (O_{P_i}-O^*_{P_i})\right] \\
        &\leq \left( \wt^{\notin\calS^*}\left[O^*_{P_i}\right]^{1/2}+ \wt^{\notin\calS^*}\left[O_{P_i}-O^*_{P_i}\right]^{1/2}\right)^2
    \end{align}
    By \Cref{thm:weight_zero_support}, we have that $\wt^{\notin\calS^*}\left[O^*_{P_i}\right] = 0$. Therefore, 
    \begin{align}
        \wt^{\notin\calS^*}[O_{P_i}] &\leq \wt^{\notin\calS^*}\left[O_{P_i}-O^*_{P_i}\right]\leq \frobd( O_{P_i},O^*_{P_i}).
    \end{align}
    This concludes the proof of this lemma.
\end{proof}

\section{Efficient Learning of \texorpdfstring{\QACZ}~ Circuit Unitaries} \label{sec:learn}
Leveraging the low-support concentration of Heisenberg-evolved \QACZ~observables, we will prove the main result of this work -- a sample and time efficient algorithm for learning $n$-output \QACZ~unitaries.
\begin{theorem}[Learning shallow circuits with many-qubit gates] \label{thm:main_result}
    Consider an unknown $n$-qubit, depth-$d$ \QACZ~circuit governed by unitary $C$. For error parameter $\epsilon = 1/\poly(n)$ and failure probability $\delta \in (0,1)$, we can learn a $2n$-qubit unitary $C_\text{sew}$ such that
    \begin{align}
        \davg(C_\text{sew}, C \otimes C^\dagger) \leq \varepsilon,
    \end{align}
    with high probability $1-\delta$. $C_\text{sew}$ can be learned with quasi-polynomial sample and time complexity.
\end{theorem}

\noindent At a high-level, the proof of this theorem consists of two main parts:
\begin{enumerate}
    \item In \Cref{sec:learn_obs}, we demonstrate that the \QACZ~circuit's Heisenberg-evolved single-qubit Pauli observables can be efficiently learned to high accuracy. This section is where most of the algorithm's novelty lies. Via classical shadow tomography, we show that an observable with $\calO(\log^d n)$-support can be learned in quasi-polynomial sample and time complexity. Then, leveraging the previously-established concentration results, we prove that this low-support learned observable is $1 /\poly(n)$-close to the true Heisenberg-evolved observable.
    \item In \Cref{sec:sew}, we leverage \cite[Section 5.2.2]{huang2024learning}'s procedure for ``sewing'' these learned Heisenberg-evolved Pauli observables into a unitary description of the circuit. Note that novel work is done to guarantee that the sewed unitary is close to the true unitary under the \emph{average-case} distance measure. (In \cite{huang2024learning}, guarantees were given according to the worst-case measure).
\end{enumerate} 
In \Cref{sec:improve_depth}, we also offer an efficient procedure to synthesize an explicit poly-logarithmic depth \QAC~circuit that implements a unitary $\frac{1}{\poly(n)}$-close to $C \otimes C^\dagger$.

\subsection{Approximate Learning of \texorpdfstring{\QACZ}~ Heisenberg-Evolved Observables} \label{sec:learn_obs}

In this section, we propose an efficient algorithm, \Cref{alg:learning_alg}, for learning an observable $\widetilde{O}^{(\support)}_{P_i}$, supported on $\support=\calO(\log^d n)$ qubits. The majority of the section will focus on proving the following learning guarantee, which establishes that $\widetilde{O}^{(\support)}_{P_i}$ is $1/\poly(n)$-close to the true \QACZ~Heisenberg-evolved Pauli observable $O_{P_i}$. 

\begin{lemma} \label{thm:learn_obs}
    Let $b \geq 2$ and $c \geq 3d$ be constants. Let $\delta \in (0,1)$ be a failure probability. For observable $O_{P_i}$,  with high probability, $1-\delta$, we can learn an approximate observable $\widetilde{O}^{(\support)}_{P_i}$ such that 
    \begin{align}
        \frobd\left( O_{P_i}, \widetilde{O}_{P_i}^{(\support)}\right) \leq \epsilon_{P_i} \leq   \frac{2}{n^b} + \frac{9d^2}{c^2 \cdot n^{c-2}}  = \frac{1}{\poly(n)}.
    \end{align}
    The sample and time complexity of this procedure are quasi-polynomial, $\calO \left( n^{\poly\log n} \cdot \log \left(1/\delta\right)\right)$.
\end{lemma}

Importantly, our proposed algorithm fundamentally differs from those of \cite{huang2024learning} and \cite{nadimpalli2023pauli}. In particular, we cannot simply apply \cite[Lemma 10]{huang2024learning} because that result assumes that the true observable to be learned is low-degree. However, in our setting, the true \QACZ~observable is not low-degree (just \emph{close to} a low-degree observable in average-case distance), meaning that the procedure's guarantees no longer hold. Meanwhile, the learning algorithm of \cite{nadimpalli2023pauli} simply leverages the low-degree concentration result to efficiently learn a degree-$\support$ Choi representation, which is close to the true single-output \QACZ~channel Choi representation. However, as will become apparent in the next section, to achieve an efficient computational complexity for the observable sewing procedure, it is critical that we learn an observable which has an $\support$-qubit \emph{support}. That is, we need to learn an observable which is not only low-degree, but also low-support. We cannot simply sample from the low-degree support, because apriori we do not actually know which qubits the observable is supported on.

\subsubsection{The Low-Support Observable Learning Algorithm}
We will now describe \Cref{alg:learning_alg} -- our proposed procedure for learning the Heisenberg-evolved single-qubit Pauli observables of a \QACZ~circuit. With this algorithm, we aim to learn, for each observable $O_{P_i}$, an approximation,
\begin{align}
    \widetilde{O}_{P_i} = \sum_{Q \in \calP^n} \widetilde{O}_{P_i}(Q) \cdot Q,
\end{align}
such that the normalized Frobenius distance is guaranteed to be small, 
\begin{align}
    \frobd \left( O_{P_i}, \widetilde{O}_{P_i}\right) = \sum_{Q \in \calP^n} \left| \widehat{O}_{P_i}(Q)-\widetilde{O}_{P_i}(Q)\right|^2 \leq \epsilon_{P_i}.
\end{align}
Note that $O_{P_i}$ could be supported on all $n$ qubits, meaning $\widehat{O}_{P_i}(Q)$ could be non-zero for all $4^n$ possible Paulis $Q$. Thus, if we were simply to try and learn all the Pauli coefficients of $O_{P_i}$, we would require a learning algorithm with exponential complexity. 

However, in \Cref{sec:heis_evolve_obs}, we saw that $O_{P_i}$ is close in distance to $O^*_{P_i}$, where $|\supp(O^*_{P_i})|\leq \support = O(\log^dn)$. This implies that $O_{P_i}$ is concentrated on the set of Paulis $S^*_\support = \{P \in \paulis: P_i =I,~\forall i \notin \supp(O^*_{P_i})\}$. Therefore, we should be able to learn a decent approximation of $O_{P_i}$ simply by learning the $S_\support^*$-truncated approximation,
\begin{align} \label{eqn:approx_trunc}
    \widetilde{O}^*_{P_i} = \sum_{Q \in S_\support^*} \widetilde{O}_{P_i}(Q) \cdot Q.
\end{align}
Note, however, that we will be learning the coefficients directly from $O_{P_i}$ and do not apriori know which qubits are contained in $\supp(O^*_{P_i})$. Therefore, we will first need to learn approximations of all the degree-$\support$ Pauli coefficients of $O_{P_i}$, i.e. 
\begin{align}
    \widetilde{O}_{P_i}(Q), \quad \forall~Q\in \calF_\support = \{P \in \pauli^n: |P| \leq \support\}.
\end{align}
Then, we will select the learned $\support$-qubit support with maximal weight. 

Formally, to describe this we will need to introduce a bit of notation. Let $S = \{L \subset [n]:|L|=\support\}$ denote the set of all possible $\support$-qubit subsets of the $n$ total qubits. For an $\support$-qubit subset $s \in S$, define the set of all Paulis supported on that set as
\begin{align}
    \calF_{\{s\}} = \left\{P_s \otimes I_{\bar{s}}~|~\forall P \in \pauli^\support\right\}.
\end{align}
Finally, denote the set of all possible $\support$-qubit supports as  
\begin{align}
    \mathfrak{F}_\support = \left\{ \calF_{\{s\}}~|~\forall s \in S\right\}.
\end{align}
Thus, the set of Paulis in an $\support$-qubit support with maximal weight is defined as
\begin{align} \label{eqn:choose_support}
    T_\support = \underset{  \calF_{\{s\}} \in \mathfrak{F}_\support}{\arg\max} \sum_{Q \in S} \left|\widetilde{O}_{P_i}(Q)\right|^2.
\end{align}
We will set any learned coefficients outside this set to zero, i.e.
\begin{align}
    \widetilde{O}^{(\support)}_{P_i}(Q) = \begin{cases}
        \widetilde{O}_{P_i}(Q), & \text{ if } Q \in T_\support, \\
        0, & \text{ otherwise}
    \end{cases}.
\end{align}
Thus, \Cref{alg:learning_alg} will learn the observable 
\begin{align}
    \widetilde{O}^{(\support)}_{P_i} = \sum_{Q \in T_\support} \widetilde{O}^{(\support)}_{P_i}(Q) \cdot Q.
\end{align}

\subsubsection{Observable Learning Guarantees}
We will now prove guarantees for \Cref{alg:learning_alg} -- namely, that the learned observable $ \widetilde{O}^{(\support)}_{P_i}$ is $1/\poly(n)$-close to the true observable $O_{P_i}$.
\begin{lemma}[Observable Learning Guarantees] \label{thm:learn_dist}
    Let $O_{P_i}$ be a \QACZ~Heisenberg-evolved Pauli observable, which is $\epsilon^*$-close to the observable $O_{P_i}^*$, with all gates of width $\geq \gaterem$ removed. Furthermore, suppose that we can learn all the  degree-$\support$ Pauli coefficients of $O_{P_i}$ to precision $\eta$, i.e.
    \begin{align} \label{eqn:approx_err}
        \left| \widehat{O}_{P_i}(Q)-\widetilde{O}_{P_i}(Q)\right| \leq \eta, \quad \forall Q \in \{P \in \pauli^n: |P|\leq \support\}.
    \end{align}
    Leveraging these learned coefficients, \Cref{alg:learning_alg} will produce a learned  observable $\widetilde{O}^{(\support)}_{P_i}$, such that
    \begin{align}
        \frobd \left( \widetilde{O}^{(\support)}_{P_i}, O_{P_i} \right) \leq 2 \cdot  4^\support \cdot \eta^2 + \epsilon^*.
    \end{align}
\end{lemma}

\begin{proof}[Proof of \Cref{thm:learn_dist}]
    The error can be decomposed as,
    \begin{align}
        \frobd \left( \widetilde{O}^{(\support)}_{P_i}, O_{P_i} \right) &= \sum_{Q \in \calP^n} \left| \widehat{O}_{P_i}(Q)-\widetilde{O}^{(\support)}_{P_i}(Q)\right|^2 \\
        &= \sum_{Q\in T_\support} \left| \widehat{O}_{P_i}(Q)-\widetilde{O}^{(\support)}_{P_i}(Q)\right|^2 + \sum_{Q\notin T_\support} \left| \widehat{O}_{P_i}(Q)-\widetilde{O}^{(\support)}_{P_i}(Q)\right|^2\\
        &\leq |T_\support|\cdot \eta^2 + \sum_{Q\notin T_\support} \left| \widehat{O}_{P_i}(Q)\right|^2.
    \end{align}
    In order to establish the desired error bound, we will upper-bound $\sum_{Q\notin T_\support} \left| \widehat{O}_{P_i}(Q)\right|^2$. Leveraging \Cref{eqn:choose_support} and the triangle inequality,
    \begin{align}
        \sum_{Q \in S_\support^*} \left| \widetilde{O}^{(\support)}_{P_i}(Q)\right|^2 &\leq \sum_{Q \in T_\support} \left| \widetilde{O}^{(\support)}_{P_i}(Q)\right|^2 \\
        & \leq \sum_{Q \in T_\support} \left(\left| \widehat{O}_{P_i}(Q)\right|^2+\left| \widetilde{O}^{(\support)}_{P_i}(Q)-\widehat{O}_{P_i}(Q)\right|^2\right) \\
        & \leq |T_\support| \cdot \eta^2 + \sum_{Q \in T_\support} \left| \widehat{O}_{P_i}(Q)\right|^2,
    \end{align}
    which implies that 
    \begin{align}
        \sum_{Q \in T_\support} \left| \widehat{O}_{P_i}(Q)\right|^2 \geq \left(\sum_{Q \in S_\support^*} \left| \widetilde{O}^{(\support)}_{P_i}(Q)\right|^2\right) - |T_\support| \cdot \eta^2.
    \end{align}
    Leveraging Parseval's (\Cref{fact:parseval}) and the fact that $\widetilde{O}^{(\support)}_{P_i}$ only has non-zero coefficients for Paulis in $T_\support$,
    \begin{align}
        \sum_{Q \notin T_\support} \left| \widehat{O}_{P_i}(Q)\right|^2 \leq |T_\support| \cdot \eta^2+\sum_{Q \notin S_\support^*} \left| \widetilde{O}^{(\support)}_{P_i}(Q)\right|^2 =|T_\support| \cdot \eta^2+\sum_{Q \notin S_\support^*, Q \in T_\support} \left| \widetilde{O}^{(\support)}_{P_i}(Q)\right|^2
    \end{align}
    Using the triangle inequality and bound from \Cref{thm:pauli_lb}, we obtain an upper-bound in terms of $\eta$ and $\epsilon^*$,
    \begin{align}
        \sum_{Q\notin T_\support} \left| \widehat{O}_{P_i}(Q)\right|^2 
        &\leq |T_\support| \cdot \eta^2+ \sum_{Q \notin S_\support^*, Q \in T_\support}\left| \widehat{O}_{P_i}(Q)-\widetilde{O}^{(\support)}_{P_i}(Q)\right|^2 + \sum_{Q \notin S_\support^*, Q \in T_\support} \left| \widehat{O}_{P_i}(Q)\right|^2 \\
        &\leq |T_\support| \cdot \eta^2+ \sum_{Q \in T_\support}\left| \widehat{O}_{P_i}(Q)-\widetilde{O}^{(\support)}_{P_i}(Q)\right|^2 + \sum_{Q \notin S_\support^*} \left| \widehat{O}_{P_i}(Q)\right|^2 \\
        & \leq |T_\support| \cdot \eta^2 +|T_\support| \cdot \eta^2 + \epsilon^* \\
        & = 2 \cdot |T_\support| \cdot \eta^2 + \epsilon^*
    \end{align}
    Noting that $T_\support$ is the set of Pauli coefficients in the support of $\support$-qubits, $|T_\support|=4^\support$, we obtain the desired result.
\end{proof}

Leveraging \Cref{thm:learn_dist}, classical shadow tomography \cite{huang2020predicting} (as described in \Cref{thm:shadow_tomog}), and some of our prior concentration results, we can now prove \Cref{thm:learn_obs}. In particular, we will now prove that the learning procedure requires quasi-polynomial sample and time complexity.
\begin{proof}[Proof of \Cref{thm:learn_obs}]
    Leveraging our bound on the distance between $O_{P_i}$ and $\widetilde{O}^{(\support)}_{P_i}$ from \Cref{thm:learn_dist}, our goal is to show that 
    \begin{align}
        \frobd\left( O_{P_i}, \widetilde{O}^{(\support)}_{P_i}\right) \leq \epsilon_{P_i} = 2 \cdot 4^\support \cdot \eta^2 + \epsilon^*  \leq \frac{1}{\poly(n)}.
    \end{align}
    
    Setting $\support = C \cdot \log^d n$ (where $C = c^d$ and $c \geq 3d$), this implies that $m=d\cdot \lfloor \frac{n}{\gaterem} \rfloor$ gates of width at least $\gaterem = c \cdot \log n$ are removed from $O_{P_i}$ to obtain $O_{P_i}^*$. Plugging this into \Cref{thm:avg_dist} implies that the distance between $O_{P_i}$ and $O_{P_i}^*$ is bounded as
    \begin{align} \label{eqn:error}
          \epsilon^* \leq \frac{9m^2}{2^{\gaterem}} \leq \frac{9\cdot d^2 \cdot \left\lfloor \frac{n}{\gaterem} \right\rfloor^2}{2^{\gaterem}} \leq \frac{9\cdot d^2 \cdot n^2}{\gaterem^2 \cdot 2^{\gaterem}} = \frac{9\cdot d^2 \cdot n^2}{c^2 \log^2 n \cdot 2^{\log(n^c)}} \leq  \frac{9d^2}{c^2n^{c-2}} 
    \end{align}
    Therefore, setting the observable learning accuracy to
    \begin{align}
        \eta^2 = \frac{1}{n^b \cdot 4^\support},
    \end{align}
    for some constant $b \geq 2$, achieves the desired error bound:
    \begin{align}
        \epsilon_{P_i} = 2 \cdot 4^\support \cdot \eta^2 + \epsilon^*\leq   \frac{2}{n^b} + \frac{9d^2}{c^2 \cdot n^{c-2}} \leq \frac{2}{n^b} + \frac{1}{n^{3d-2}} = \frac{1}{\poly(n)}.
    \end{align}

    Now all that remains is to prove the sample and computational complexity. Denote the set of $n$-qubit Paulis of degree $\leq \support$ as
    \begin{align}
        \calF_\support = \{ P \in \pauli^n : |P| \leq \support\}.
    \end{align}
    Via classical shadow tomography, with the state $\rho = O_{P_i}$ and the set of bounded-degree Pauli observables $\calF_\support$, we can produce an $\eta$-estimate $\widetilde{O}_{P_i} (Q)$ of 
    \begin{align}
        \widehat{O}_{P_i} (Q) = \Tr(Q \cdot O_{P_i})
    \end{align}
    for each $Q \in \calF_\support$. Since the Pauli observables are all of degree at most $\ell$, by \Cref{thm:shadow_tomog}, the sample complexity of this classical shadow tomography procedure is 
    \begin{align}
        \calO \left( \frac{3^\support}{\eta^2} \log \left(\frac{n^\support}{\delta}\right)\right) = \calO \left( 3^\support \cdot n^b \cdot 4^\support \cdot \log \left(\frac{n^\support}{\delta}\right)\right) = \calO \left( n^b \cdot 4^{2\support} \cdot \log \left(\frac{n^\support}{\delta}\right)\right).
    \end{align}
    Furthermore, since the size of $\calF_\support$, or the total number of Pauli of degree $\leq \support$, is upper-bounded as
    \begin{align}
        |\calF_\support| = \sum_{k=1}^\support 3^k \cdot {n \choose k}  \leq \calO (n^\support),
    \end{align}
    the computational complexity is 
    \begin{align}
        \calO \left( n^{\support+b} \cdot 4^{2\support} \cdot \log \left(\frac{n^\support}{\delta}\right)\right). 
    \end{align}

    As described in \Cref{alg:learning_alg}, to produce the  approximation $\widetilde{O}^{(\support)}_{P_i}$, we also need to find the set of Paulis supported on $\ell$ qubits with maximal weight and set all other Pauli coefficients to zero. Since there are ${n \choose \ell}$ possible supports of size $\ell$ and calculating the weight of each support involves summing over $3^\ell$ different Pauli coefficients, the computational complexity of finding the set $T_\ell$ is quasi-polynomial, i.e.
    \begin{align}
        {n \choose \ell} \cdot 3^\ell \leq \calO(n^\ell).
    \end{align}

    To achieve the final, explicit sample and computational complexity, plug in the value $\support = C \cdot \log^d n$, where $C = c^d$ and $c \geq 3d$. Thus, the overall sample complexity is quasi-polynomial 
    \begin{align}
        \calO \left( n^b \cdot 4^{2\support} \cdot \log \left(\frac{n^\support}{\delta}\right)\right) &= \calO \left( n^b \cdot 2^{4 C  \log^d n} \cdot \log \left(\frac{n^{C \log^d n}}{\delta}\right)\right) \\
         &= \calO \left( n^b \cdot 2^{4 C  \log^d n} \cdot\left(C \log^d n \cdot \log n -\log \delta\right)\right) \\
         &= \calO \left( n^b \cdot 2^{4 c^d  \log^d n} \cdot\left( \log^{d+1} n +\log (1/\delta)\right)\right) \\
        &= \calO \left( 2^{\poly\log n} \cdot \log \left(1/\delta\right)\right),
    \end{align}
    and the computational complexity is also quasi-polynomial
    \begin{align}
        \calO \left( n^{\support+b} \cdot 4^{2\support} \cdot \log \left(\frac{n^\support}{\delta}\right)+n^\ell\right) &= \calO \left( n^{\support+b} \cdot 4^{2\support} \cdot \log \left(\frac{n^\support}{\delta}\right)\right) \\
        &= \calO \left( n^{c^d  \log^d n+b} \cdot 2^{4 c^d  \log^d n} \cdot\left( \log^{d+1} n +\log (1/\delta)\right)\right) \\
        &= \calO \left( n^{\poly\log n} \cdot \log \left(1/\delta\right)\right).
    \end{align}
\end{proof}

\subsection{Sewing \texorpdfstring{\QACZ}~ Heisenberg-Evolved Observables} \label{sec:sew}
We will now use a procedure analogous to that of \cite[Section 5.2.2]{huang2024learning} to sew these Heisenberg-evolved Pauli observables and project them onto a unitary which is $1/\poly(n)$-close to the true \QACZ~circuit unitary, with respect to the average-case distance.

Analogous to \cite[Lemma 9]{huang2024learning}, we begin by showing that the error in the sewing procedure can be upper-bounded by the sum of the learned observables' learning error (from the last section). Importantly, note that our proof differs from that of \cite{huang2024learning}, since we are leveraging an average-case instead of a worst-case distance measure.
\begin{lemma}[Observable Sewing Guarantees] \label{thm:csew_bound}
    Suppose $C$ is an $n$-qubit \QACZ~circuit, which has a set of Heisenberg-evolved observables $\left\{O_{P_i}\right\}_{i,P}$, corresponding to each of the $3n$ possible single-qubit Paulis $P_i$. Let $\left\{\widetilde{O}^{(\support)}_{P_i}\right\}_{i,P}$ denote the set of learned observables, which are each at most $\epsilon_{P_i}$-far from the respective true observable. Construct the unitary 
    \begin{align}
        C_\text{sew}  := \textnormal{SWAP}^{\otimes n} \prod_{i=1}^n \left[\proj_U\left(\frac{1}{2} I\otimes I+\frac{1}{2}\sum_{P\in\{X,Y,Z\}} \widetilde{O}_{P_i}^{(\support)} \otimes P_i\right) \right],
    \end{align}
    by ``sewing'' the learned observables, where $\proj_U$ is the projection onto the unitary minimizing Frobenius norm, as defined in \Cref{eqn:proj_u}, and $\textnormal{SWAP}^{\otimes n}$ swaps the first and last $n$ qubits. The average-case distance between $C_\text{sew}$ and $C \otimes C^\dagger$ is at most
    \begin{align}
        \davg(C_\text{sew},C \otimes C^\dagger) \leq \frac{1}{2}~\sum_{i=1}^n \sum_{P\in\{X,Y,Z\}} \epsilon_{P_i}.
    \end{align}
\end{lemma}
\begin{proof}[Proof of \Cref{thm:csew_bound}]
    To begin, note that $C\otimes C^\dagger$, can be decomposed in terms of the true observables as 
    \begin{align}
        C \otimes C^\dagger = \textnormal{SWAP}^{\otimes n} \prod_{i=1}^n \left[\frac{1}{2} I\otimes I+\frac{1}{2}\sum_{P\in\{X,Y,Z\}} O_{P_i} \otimes P_i \right].
    \end{align}
    Also, define the following three matrices: 
    \begin{align}
        V_i &= \frac{1}{2} I\otimes I+\frac{1}{2}\sum_{P\in\{X,Y,Z\}} O_{P_i} \otimes P_i\\
        \widetilde{W}_i &= \frac{1}{2} I\otimes I+\frac{1}{2}\sum_{P\in\{X,Y,Z\}} \widetilde{O}_{P_i}^{(\support)} \otimes P_i\\
        W_i &= \proj_U\left(\widetilde{W}_i\right).
    \end{align}
    
    Noting that $C_\text{sew}$ and $C \otimes C^\dagger$ implement unitary channels, we can leverage \Cref{fact:davg_ub} and \Cref{fact:gpid_ub} to upper-bound the average gate fidelity of the channels by the normalized Frobenius distance of the corresponding unitaries, 
    \begin{align}
        \davg(C_\text{sew},C \otimes C^\dagger) \leq \gpid(C_\text{sew}, C \otimes C^\dagger) \leq \frobd(C_\text{sew}, C \otimes C^\dagger).
    \end{align}
    Leveraging the facts that the Frobenius norm is unitary invariant and satisfies the triangle inequality, we perform the following hybrid argument:
    \begin{align}
        \davg(\calC_\text{sew},\calC \otimes \calC^\dagger) &\leq \frobd(C_\text{sew}, C \otimes C^\dagger) \\
        &= \frobd\left( S \prod_{i=1}^n W_i, S \prod_{i=1}^n V_i\right) \\
        &= \frobd \left( \prod_{i=1}^n W_i, \prod_{i=1}^n V_i\right) \\
        &\leq \frobd\left( \prod_{i=1}^n W_i, \prod_{i=1}^{n-1} W_i V_n\right)+\frobd \left( \prod_{i=1}^{n-1} W_i V_n, \prod_{i=1}^n V_i\right) \\
        &\leq \frobd \left( W_n,  V_n\right)+\frobd \left( \prod_{i=1}^{n-1} W_i , \prod_{i=1}^{n-1} V_i\right) \\
        &\quad \quad \quad \vdots \\
        & \leq  \sum_{i=1}^n \frobd \left( W_i,  V_i\right). \label{eqn:davg_bound}
    \end{align}
    Upper-bounding each $\frobd \left( W_i,  V_i\right)$ term as
    \begin{align} \label{eqn:trig_ub}
       \frobd \left( W_i,  V_i\right) \leq \frobd \left( W_i,  \widetilde{W}_i\right)+\frobd \left( \widetilde{W}_i,  V_i\right),
    \end{align}
    we are now interested in the values of $\frobd \left( W_i,  \widetilde{W}_i\right)$ and $\frobd \left( \widetilde{W}_i,  V_i\right)$. However, leveraging the definition of $\proj_U$, $\frobd \left( W_i,  \widetilde{W}_i\right)$ can be expressed in terms of $\frobd \left( \widetilde{W}_i,  V_i\right)$, as
    \begin{align}
        \frobd\left(\widetilde{W}_i, W_i\right) &= \frac{1}{2^{2n}} \left\|\widetilde{W}_i-W_i\right\|_F^2 
        = \frac{1}{2^{2n}} \min_{U\in\mathcal{U}(2^{2n})}\left\|\widetilde{W}_i-U\right\|_F^2 
        \leq \frac{1}{2^{2n}} \left\|\widetilde{W}_i-V_i\right\|_F^2 
        = \frobd\left(\widetilde{W}_i, V_i\right).
    \end{align}
    Therefore, \Cref{eqn:trig_ub} simplifies to 
    \begin{align} 
        \frobd \left( W_i,  V_i\right) \leq 2\cdot \frobd \left( \widetilde{W}_i,  V_i\right). \label{eqn:davg_bound2}
    \end{align}
    Furthermore, we can upper-bound $\frobd \left( \widetilde{W}_i,  V_i\right)$, as
    \begin{align}
        \frobd \left( \widetilde{W}_i,  V_i\right) &= \frac{1}{2^{2n}} \left\|\widetilde{W}_i - V_i\right\|_F^2 \\
        & \leq \frac{1}{2^{2n}} \left\|\left(\frac{1}{2} I\otimes I+\frac{1}{2}\sum_{P\in\{X,Y,Z\}} \widetilde{O}_{P_i}^{(\support)} \otimes P_i\right) - \left(\frac{1}{2} I\otimes I+\frac{1}{2}\sum_{P\in\{X,Y,Z\}} O_{P_i} \otimes P_i\right)\right\|_F^2 \\
        & \leq \frac{1}{2^{2n}} \left\|\frac{1}{2}\sum_{P\in\{X,Y,Z\}} \left( \widetilde{O}_{P_i}^{(\support)}  -  O_{P_i} \right)\otimes P_i\right\|_F^2 \\
         & \leq \frac{1}{2^{2n}} \cdot \frac{1}{4}\Tr \left(\sum_{P\in\{X,Y,Z\}} \left( \widetilde{O}_{P_i}^{(\support)}  -  O_{P_i} \right)^\dagger\otimes P_i\cdot \sum_{Q\in\{X,Y,Z\}} \left( \widetilde{O}_{P_i}^{(\support)}  -  O_{Q_i} \right)\otimes Q_i\right) \\
         & = \frac{1}{2^{2n}} \cdot \frac{1}{4}\sum_{P,Q\in\{X,Y,Z\}}\Tr \left( \left( \widetilde{O}_{P_i}^{(\support)}  -  O_{P_i} \right)^\dagger \left( \widetilde{O}_{P_i}^{(\support)}  -  O_{Q_i} \right) \otimes P_i Q_i\right) \\
         & = \frac{1}{2^{2n}} \cdot \frac{1}{4}\sum_{P,Q\in\{X,Y,Z\}}\Tr \left( \left( \widetilde{O}_{P_i}^{(\support)}  -  O_{P_i} \right)^\dagger \left( \widetilde{O}_{P_i}^{(\support)}  -  O_{Q_i} \right) \right)\Tr\left(P_i Q_i\right) \\
         & = \frac{1}{2^{2n}} \cdot \frac{1}{4}\sum_{P,Q\in\{X,Y,Z\}}\Tr \left( \left( \widetilde{O}_{P_i}^{(\support)}  -  O_{P_i} \right)^\dagger \left( \widetilde{O}_{P_i}^{(\support)}  -  O_{Q_i} \right) \right) \cdot 2^n \delta(P_i=Q_i) \\
         & = \frac{1}{2^{n}} \cdot\frac{1}{4}\sum_{P\in\{X,Y,Z\}}\Tr \left( \left( \widetilde{O}_{P_i}^{(\support)}  -  O_{P_i} \right)^\dagger \left( \widetilde{O}_{P_i}^{(\support)}  -  O_{P_i} \right) \right) \\
         & = \frac{1}{4}\sum_{P\in\{X,Y,Z\}} \frac{1}{2^{n}} \left\| \widetilde{O}_{P_i}^{(\support)}  -  O_{P_i} \right\|_F^2 \\
         & = \frac{1}{4}\sum_{P\in\{X,Y,Z\}} \epsilon_{P_i}. \label{eqn:davg_bound3}
    \end{align}
    Therefore, combining our bounds from \Cref{eqn:davg_bound}, \Cref{eqn:davg_bound2}, and \Cref{eqn:davg_bound3}, we obtain the desired upper-bound:
    \begin{align}
        \davg(C_\text{sew}, C \otimes C^\dagger) \leq  \sum_{i=1}^n \frobd \left( W_i,  V_i\right) \leq  \sum_{i=1}^n 2\cdot \frobd \left( \widetilde{W}_i,  V_i\right) \leq \frac{1}{2} \sum_{i=1}^n \sum_{P\in\{X,Y,Z\}} \epsilon_{P_i}.
    \end{align}
\end{proof}

With these results establishing efficient learning and sewing of \QACZ~Heisenberg-evolved single-qubit Pauli observables, we can now prove the main result of the section, \Cref{thm:main_result}.

\begin{proof}[Proof of \Cref{thm:main_result}]
    Leveraging \Cref{alg:learning_alg} and \Cref{thm:learn_obs}, we can learn the set of observables $\{\widetilde{O}_{P_i}^{(\support)}\}_{i,P}$ such that 
    \begin{align}
        \frobd \left( O_{P_i},\widetilde{O}^{(\support)}_{P_i}\right) \leq \epsilon_{P_i} \leq \frac{2}{n^b} + \frac{9d^2}{c^2 \cdot n^{c-2}}
    \end{align}
    for all $3n$ single-qubit Paulis $P_i$, with quasi-polynomial sample and computational complexity. Leveraging the result of \Cref{thm:csew_bound}, we can then ``sew'' these learned observables into the unitary 
    \begin{align} \label{eqn:sew}
        C_\text{sew} := \textnormal{SWAP}^{\otimes n} \prod_{i=1}^n \left[\proj_U\left(\frac{1}{2} I\otimes I+\frac{1}{2}\sum_{P\in\{X,Y,Z\}} \widetilde{O}_{P_i}^{(\support)} \otimes P_i\right) \right],
    \end{align}
    which, for constants $b \geq 2$ and $c \geq 3d$, satisfies
    \begin{align} \label{eqn:sew_comp}
        \davg(\calC_\text{sew},\calC \otimes \calC^\dagger) \leq \frac{1}{2} \sum_{i=1}^n \sum_{P\in\{X,Y,Z\}} \epsilon_{P_i} \leq \frac{3n}{2} \cdot \underset{P_i}{\arg\max}~\epsilon_{P_i} \leq  \frac{3}{n^{b-1}} + \frac{9d^2}{c^2 \cdot n^{c-3}} = \frac{1}{\poly (n)},
    \end{align}
    meaning that $C_\text{sew}$ is $1/\poly(n)$-close to $C \otimes C^\dagger$ in average-case distance.

    All that remains is to verify that the computational complexity of the sewing procedure of \Cref{eqn:sew} is in fact quasi-polynomial. Note that the construction of $C_\text{sew}$ requires computing 
    \begin{align}
        \widetilde{U}_i = \proj_U\left(\frac{1}{2} I\otimes I+\frac{1}{2}\sum_{P\in\{X,Y,Z\}} \widetilde{O}_{P_i}^{(\support)} \otimes P_i\right)
    \end{align}
    $n$ times. As described in \Cref{fact:project}, computing $\proj_U$ reduces to computing the singular value decomposition of the matrix. In general, computing the SVD of a $2n$-qubit matrix has exponential complexity $\calO(2^{6n})$. However, in \Cref{alg:learning_alg} we specifically imposed that the learned observable $\widetilde{O}_{P_i}^{(\support)}$ only have support on $\support= \calO(\log^d n)$ qubits. Since $P_i$ only has support on 1 qubit, the total support of the matrix $\frac{1}{2} I\otimes I+\frac{1}{2}\sum_{P\in\{X,Y,Z\}} \widetilde{O}_{P_i}^{(\support)} \otimes P_i$ is $\support+1$ qubits. Therefore, we only need to compute the SVD of the sub-matrix corresponding to the non-trivial support, which has quasi-polynomial computational complexity $\calO(2^{\poly \log n})$. Therefore, the total sewing procedure involves performing $n$ of these SVDs, which is still quasi-polynomial complexity.
\end{proof}

\subsection{Learning \texorpdfstring{\QACZ}~ with Optimized Depth} \label{sec:improve_depth}
Now that we have learned a unitary $C_\text{sew}$ which is close in average-case distance to $C \otimes C^\dagger$, we will explore circuit-synthesis procedures that implement unitaries close to $C_\text{sew}$. In other words, we will look into \emph{proper learning} of the \QACZ~circuit. We will show that, while a naive compilation procedure would produce a worst-case quasi-polynomial-depth circuit, we can generate in quasi-polynomial time an explicit (poly-logarithmic depth) \QAC~circuit $C^*_\text{sew}$ that implements a unitary $\frac{1}{\poly(n)}$-close to $C_\text{sew}$. In particular, we will prove the following theorem.

\begin{theorem} \label{thm:uni_synthesis}
    Given a \QACZ~circuit $C$, there exists a quasi-polynomial time algorithm to learn a \QAC~circuit implementing unitary $C_\text{sew}^*$ such that
    \begin{align}
        \frobd \left(C_\text{sew}^*, C \otimes C^\dagger \right) \leq \frac{1}{\poly(n)}.
    \end{align}
\end{theorem}

\noindent Throughout this section, we will use the following notation:
\begin{align}
    V_i &= \frac{1}{2} I\otimes I+\frac{1}{2}\sum_{P\in\{X,Y,Z\}} O_{P_i} \otimes P_i,\\
    V^*_i &= \frac{1}{2} I\otimes I+\frac{1}{2}\sum_{P\in\{X,Y,Z\}} O^*_{P_i} \otimes P_i,\\
    \widetilde{W}_i &= \frac{1}{2} I\otimes I+\frac{1}{2}\sum_{P\in\{X,Y,Z\}} \widetilde{O}_{P_i}^{(\support)} \otimes P_i,\\
    W_i &= \proj_U\left(\widetilde{W}_i\right).
\end{align}
Furthermore, by the sewing procedure of \cite{huang2024learning}, if $O_{P_i}$ corresponds to the \QACZ~circuit $C$ and $L_i$ is the light-cone of $C$ with respect to measurement qubit $i$, then
\begin{align} 
    V_i &= \frac{1}{2} I\otimes I+\frac{1}{2}\sum_{P\in\{X,Y,Z\}} O_{P_i} \otimes P_i \\
    &= \frac{1}{2}\sum_{P\in\{I,X,Y,Z\}} C^\dagger P_i C \otimes P_i \\
    &= \frac{1}{2}\sum_{P\in\{I,X,Y,Z\}} L_i^\dagger P_i L_i \otimes P_i \\
    &= (L_i^\dagger \otimes I)\left(\frac{1}{2}\sum_{P\in\{I,X,Y,Z\}}  P_i \otimes P_i \right) (L_i \otimes I) \\
    &= L_i^\dagger S_i L_i,
\end{align}
where $S_i$ denotes the SWAP operation between the $i^\text{i}$ and $(n+i)^\text{th}$ qubits.
Similarly, if $O_{P_i}^*$ corresponds to the \QACZ~circuit $\widetilde{C}$ (with large $CZ$ gates removed) and $\widetilde{L}_i$ is the light-cone of $\widetilde{C}$ with respect to measurement qubit $i$, then
\begin{align}
    V^*_i &= \widetilde{L}_i^\dagger S_i \widetilde{L}_i.
\end{align}

\subsubsection{Naive Implementation}

We will begin by evaluating the worst-case circuit depth of a naive compilation procedure for our learned unitary. Our naive implementation will leverage the following standard fact about the complexity of unitary synthesis.

\begin{fact}[\cite{huang2024learning} Fact 4] \label{fact:unitary_synthesis}
    Given a unitary $U$, which acts on $k$ qubits, there is an algorithm that outputs a circuit (acting on $k$ qubits) that consists of at most $4^k$ two-qubit gates, which exactly implements the unitary $U$, in time $2^{\calO(k)}$.
\end{fact}

Because any two-qubit gate can be generated by a constant number of single-qubit gates and $CZ$ gates, we can immediately obtain the following fact from the above.

\begin{fact}[Adapted from \cite{huang2024learning} Fact 4] \label{fact:unitary_synthesis-AC-circuit}
    Given a unitary $U$, which acts on $k$ qubits, there is an algorithm that outputs a circuit (acting on $k$ qubits) that consists of at most $4^k$ single-qubit gates or many-qubit $CZ$ gates, which exactly implements the unitary $U$, in time $2^{\calO(k)}$.
\end{fact}

For each qubit measurement index $i$, we will consider a naive circuit-synthesis procedure that simply compiles the individual projected unitaries $W_i$ into a circuit with gates acting on at most 2 qubits and then sews these compiled circuits into $C_\text{sew}$, as specified in \Cref{eqn:sew}, via some arbitrary ordering. Since the support of each unitary $W_i$ is $\ell = \poly\log(n)$, by \Cref{fact:unitary_synthesis}, it would be compiled into a circuit acting on $\ell = \poly\log(n)$ qubits consisting of up to $4^{\poly\log(n)}$ gates. 

Stitching together all $3n$ of these circuits in an arbitrary order results in a circuit with quasi-polynomial depth.
This is substantially deeper than the constant depth of the true \QACZ~circuit that we aimed to learn. Thus, we will now show how improved compilation and ordering of the unitaries in the sewing procedure can reduce this depth down to poly-logarithmic, while still only requiring quasi-polynomial computational complexity.

\subsubsection{Improved Compilation}
We will now describe an improved compilation procedure, that reduces the circuit-synthesis depth for unitaries $W_i$ from $\calO(4^{\poly\log n})$ to constant-depth $d$.  
\begin{theorem} \label{thm:improve_comp}
    The learned unitary $W_i$ can be compiled  into a \QACZ~circuit, governed by unitary $C_i^*$, that is supported on $\calO(\log^d n)$ qubits and such that
    \begin{align}
        \frobd \left( W_i, C_i^* \right) \leq \frac{1}{\poly(n)}.
    \end{align}
    The computational complexity of this compilation procedure is quasi-polynomial.
\end{theorem}

In order to prove this result, we will first show that the learned unitaries $W_i$ are close to the $\calO(\log^d n)$-support \QACZ~circuits governed by unitary $V_i^*=\widetilde{L}_i^\dagger S_i \widetilde{L}_i$.

\begin{lemma} \label{thm:dist_sew}
    Let $O_{P_i}$ be a Heisenberg-evolved observable of the \QACZ~circuit $C$ and $\widetilde{O}_{P_i}$ be a Heisenberg-evolved observable of some other circuit $\widetilde{C}$ such that
    \begin{align}
        \frobd\left(O_{P_i},\widetilde{O}_{P_i}\right) \leq \epsilon, \quad \forall i \in [n],~P \in \{X,Y,Z\}.
    \end{align}
    Then,
    \begin{align}
        \frobd \left( \frac{1}{2} \sum_{P\in\{X,Y,Z\}}O_{P_i} \otimes P_i,~\frac{1}{2} \sum_{P\in\{X,Y,Z\}}\widetilde{O}_{P_i}\otimes P_i\right) \leq \epsilon
    \end{align}
\end{lemma}
\begin{proof}
    Leveraging Cauchy-Schwarz and the fact that $\|A\otimes B\|_F^2 = \|A\|_F^2\| B\|_F^2$,
    \begin{align}
        \frobd &\left( \frac{1}{2} \sum_{P\in\{I,X,Y,Z\}} O_{P_i} \otimes P_i,~\frac{1}{2} \sum_{P\in\{I,X,Y,Z\}} \widetilde{O}_{P_i}\otimes P_i\right) \\
        & \leq  \frac{1}{2^n} \left\| \frac{1}{2} \sum_{P\in\{I,X,Y,Z\}} (O_{P_i} -\widetilde{O}_{P_i})\otimes P_i\right\|_F^2 \\
        &\leq  \frac{1}{4}   \sum_{P\in\{I,X,Y,Z\}} \frac{1}{2^n} \left\|(O_{P_i} -\widetilde{O}_{P_i})\otimes P_i\right\|_F^2 \\
        &\leq  \frac{1}{4}   \sum_{P\in\{I,X,Y,Z\}} \frac{1}{2^n} \left\|(O_{P_i} -\widetilde{O}_{P_i})\right\|_F^2 \left\| P_i\right\|_F^2 \\
        &=  \frac{1}{4}   \sum_{P\in\{X,Y,Z\}} \frobd \left(O_{P_i},\widetilde{O}_{P_i}\right) \\
        & \leq \frac{3}{4} \cdot \epsilon.
    \end{align}
    This concludes the proof.
\end{proof}

\begin{corollary} \label{thm:bound_dist}
    The learned unitaries $W_i$ are $1/\poly(n)$-close to the unitaries $V_i^*=\widetilde{L}_i^\dagger S_i \widetilde{L}_i$, 
    \begin{align}
        \frobd \left( W_i,V_i^*\right) \leq \frac{1}{\poly (n)}.
    \end{align}
\end{corollary}
\begin{proof}
    By \Cref{thm:learn_obs}, 
    \begin{align}
        \frobd \left(\widetilde{O}^{(\support)}_{P_i},O_{P_i}\right) \leq \epsilon_{P_i} = \frac{1}{\poly(n)},    
    \end{align}
    and, by \Cref{thm:avg_dist},
    \begin{align}
        \frobd \left( \widetilde{O}^*_{P_i},O_{P_i}\right) \leq \epsilon^* = \frac{1}{\poly(n)}.    
    \end{align}
    By triangle inequality and \Cref{thm:dist_sew}, we have that
    \begin{align}
        &\frobd \left( \widetilde{W}_i, V^*_i \right) \\
        & \leq  \frobd \left( \widetilde{W}_i, V_i \right)+ \frobd \left( V_i, V_i^* \right)\\
        & \leq \frac{1}{2^n} \left\| \frac{1}{2} \sum_{P\in\{X,Y,Z\}}\left(\widetilde{O}_{P_i}^{(\support)}\otimes P_i - O_{P_i}\otimes P_i\right) \right\|_F^2 + \frac{1}{2^n} \left\| \frac{1}{2} \sum_{P\in\{X,Y,Z\}}\left(O_{P_i}\otimes P_i - O_{P_i}^*\otimes P_i\right) \right\|_F^2\\
        & \leq \epsilon_{P_i}+\epsilon^* \leq \frac{1}{\poly(n)}.
    \end{align}
    Furthemore, since $\proj_U$ is the projection onto the unitary minimizing the Frobenius norm, for $W_i = \proj_U(\widetilde{W}_i)$, it must be true that
    \begin{align}
        \frobd \left( \widetilde{W}_i, W_i \right) \leq \frobd \left( \widetilde{W}_i, V^*_i \right) \leq \frac{1}{\poly(n)}.
    \end{align}
    Therefore, by triangle inequality we obtain the desired result,
    \begin{align}
        \frobd \left( W_i, V_i^*\right) \leq \frobd \left( \widetilde{W}_i, W_i \right) + \frobd \left( \widetilde{W}_i, V^*_i \right)  \leq \frac{1}{\poly (n)}.
    \end{align}
\end{proof}
This implies that there must exist a circuit of the form $V_i^* = \widetilde{L}_i^\dagger S_i \widetilde{L}_i$ that is $1/\poly(n)$-close to each of the learned unitaries $W_i$. Note that since the swap gate $S_i$ can be implemented in \QACZ, the circuit  $V_i^*$ is contained \QACZ. This implies that, to find a circuit $1/\poly(n)$-close to $W_i$, rather than searching over all possible \QACZ~architectures, we can restrict our search to \QACZ~circuits of the form  of $V_i^*$, or more precisely \QACZ~circuits in the lightcone $\widetilde{L}_i$. Thus, we will now show how to efficiently construct an $\epsilon$-net over \QACZ~circuits of depth-$d$, with $\calO(\log^d n)$ support, that is guaranteed to contain a \QACZ~circuit $1/\poly(n)$-close to the learned unitary $W_i$.



\begin{lemma} \label{thm:e_net}
    Let $\calC^*$ be the class of all depth-$d$ \QACZ~circuits with $CZ$ gates acting on at most $\gaterem=\calO(\log n)$ qubits and supported on  $\calO(\log^d n)$ qubits. $\calC^*$ has a $1/\poly(n)$-net, denoted $\calN_{1/\poly(n)} (\calC^*)$, of quasi-polynomial size that can be constructed in quasi-polynomial time.
\end{lemma}
\begin{proof}
    Recall that the general structure of a \QACZ~circuit is alternating layers of $CZ$ gates and layers of arbitrary single-qubit gates. Therefore, we can construct an $\epsilon$-net for $\calC^*$ by first enumerating all possible architectures (i.e. placements of $CZ$ gates) and then, for each architecture, creating an $\epsilon'$-net for each possible SU(2) gate.
    
    We will begin by enumerating all possible \QACZ~architectures, i.e. configurations of $CZ$ gates of width at most $\gaterem = \calO(\log n)$ acting on $\support=\calO(\log^d n)$ qubits across $d$ layers. We begin thinking of a given layer of the architecture as a graph with $\support$ vertices, corresponding to each of the qubits in the support. Within this graph framework, if a set of vertices are contained in a $k$-clique this means the corresponding qubits are acted upon by a $CZ_k$ gate. Therefore, to enumerate the total number of distinct $CZ$ configurations in the layer, we simply need to enumerate the number of distinct graphs comprised of cliques of size at most $\gaterem$. This is trivially upper-bounded by the number of distinct subgraphs of the complete graph on $\support$ nodes, which is quasi-polynomial, i.e.
    \begin{align}
        2^{{\support \choose 2}} \leq 2^{\support^2} = 2^{\poly\log(n)}.
    \end{align}
    Since the circuits are depth-$d$, where $d$ is constant, the total number of $CZ$ configurations across the whole circuit is $d \cdot 2^{\poly\log(n)}$, which is also quasi-polynomial.

    As previously mentioned, between the layers of $CZ$ gates are layers of arbitrary single-qubit gates. In total, there are at most $d \cdot \support \leq \calO(\log^d n)$ single-qubit gates. By a standard hybrid argument, it can be shown that the error propagation of the SU(2) $\epsilon'$-net is additive both within and across layers of the single qubit gates. Therefore, to achieve an overall $ 1/\poly(n)$-net,
    \begin{align}
        d \cdot \support \cdot \epsilon' \leq \poly\log (n) \cdot \epsilon' \leq \frac{1}{\poly(n)},
    \end{align}
    which implies that $\epsilon'\leq 1/\poly(n)$. An $\epsilon'$-net for $SU(2)$ can be constructed with $\left(\frac{c_0}{\epsilon'}\right)^{c_1}$ elements, which for $\epsilon'=1/\poly(n)$ is polynomial size. 
    
    Therefore, the total size of the net is the size of the $SU(2)$ $\epsilon'$-net times the total number of single-qubit gates times the total number of architectures, which is bounded as
    \begin{align}
        \poly(n) \cdot d \cdot \poly\log(n) \cdot 2^{\poly\log(n)} \leq \calO(2^{\poly\log(n)}),
    \end{align}
    and therefore quasi-polynomial.
\end{proof}

We will now show that we can efficiently find an element of the $\epsilon$-net that is $1/\poly(n)$-close to $W_i$. Combining this with the prior results of the section, we achieve a simple proof of \Cref{thm:improve_comp}.

\begin{proof}[Proof of \Cref{thm:improve_comp}]
    By using a brute-force search procedure, we can iterate through the quasi-polynomial elements of the $\epsilon$-net described in \Cref{thm:e_net} to find the element  
    \begin{align}
        L^*_i = \underset{L \in \calN_{1/\poly(n)} (\calC^*)}{\arg\min} \frobd \left( W_i, L^\dagger S_i L \right).
    \end{align}
    Since the swap gate $S_i$ can be implemented in constant-depth in \QACZ, we thus have a constant-depth circuit implementation of the unitary $C_i^* = (L^*_i)^\dagger S_i L^*_i$.
    By \Cref{thm:bound_dist}, 
    \begin{align}
        \frobd \left( W_i, C_i^* \right) \leq \frobd \left( W_i, \widetilde{L}_i^\dagger S_i \widetilde{L}_i \right) = \frobd \left( W_i, V_i^*\right) \leq \frac{1}{\poly (n)},
    \end{align}
    meaning that $C_i^*$ is $1/\poly(n)$-close to $W_i$, as desired.
\end{proof}

\subsubsection{Improved Ordering}
Leveraging the improved compilation result, we will now propose an improved ordering for the sewing procedure. This will enable us to construct a \QAC~circuit which is $1/\poly(n)$-close to $C_\text{sew}$, thereby proving \Cref{thm:uni_synthesis}.

To begin, we show how the constant-depth learned circuits for $C_i^*$ can be sewed into a worst-case poly-logarithmic depth circuit. Note that our proof approach is similar to that of \cite[Lemma 13]{huang2024learning}. However, the circuit for each $C_i^*$ has poly-logarithmic support (whereas those of \cite{huang2023learning} have constant support), meaning we only achieve poly-logarithmic depth (instead of constant depth). 

\begin{lemma}(Sewing into a poly-logarithmic depth circuit) \label{thm:sew_order} Given $3n$ learned observables $\widetilde{O}_{P_i}^{(\support)}$, such that for any qubit $i$, $\left| \bigcup_P \supp\left(\widetilde{O}_{P_i}^{(\support)}\right) \right| = \calO (\poly\log(n))$ and there are only $\poly\log(n)$ qubits $j$ such that 
\begin{align}
    \supp\left(\widetilde{O}_{P_i}^{(\support)}\right) \cap \supp\left(\widetilde{O}_{P_j}^{(\support)}\right) \neq \varnothing.
\end{align}
There exists a sewing ordering for $C_\text{sew}$, as defined in \Cref{eqn:sew}, such that it can be implemented by a $\poly\log(n)$-depth quantum circuit. The computational complexity for finding this sewing order is polynomial, i.e. $\calO(n \log^d n)$.
\end{lemma}
\begin{proof}[Proof of \Cref{thm:sew_order}]
    Defining $A(i) = \bigcup_P \supp\left(\widetilde{O}_{P_i}^{(\support)}\right)$, then $\supp\left(C^*_{i}\right) \subseteq A(i) \cup \{n+i\}$, where $C_i^*$ is the learned constant-depth circuit from \Cref{thm:improve_comp}. 
    
    Now, consider an $n$-node graph (where each node represents one of the $n$ qubits), such that each pair $(i,j)$ of nodes is connected by an edge if 
    \begin{align}
        A(i) \cap A(j) \neq \varnothing.
    \end{align}
    The graph only has $\calO(n \log^d n)$ many edges and can be constructed as an adjacency list in time $\calO(n \log^d n)$. Since the size of the support $A(i)$ is poly-logarithmic, the graph has poly-logarithmic degree. Thus, we can use a $\calO(n\log^d n)$-time greedy graph coloring algorithm to color the graph using $\chi = \calO(\log^d n)$ colors. For each node $i$, let $c(i)$ denote the color labeled from 1 to $\chi$. 
    
    We can modify the arbitrary sewing order of the $3n$ observables $\widetilde{O}_{P_i}^{(\support)}$ in \Cref{eqn:sew} to the ordering given by this greedy graph coloring, where we order from the smallest to the largest color. By the definition of graph coloring, for any pair $(i,j)$ of qubits with the same color,
    \begin{align}
        A(i) \cap A(j) = \varnothing.
    \end{align}
    Therefore, for each color $c'$, we can implement the $2n$-qubit unitary
    \begin{align}
        \prod_{i:c(i)=c'} C^*_i
    \end{align} 
    via the constant-depth quantum circuits $C^*_i$. Since there are at most a poly-logarithmic number of colors, with the color-based ordering, $C_\text{sew}$ will be poly-logarithmic depth in the worst-case.
\end{proof}

Combining this improved ordering result with the improved compilation result of the last section, we can now prove \Cref{thm:uni_synthesis}.

\begin{proof}[Proof of \Cref{thm:uni_synthesis}]
    For each qubit $i$, by \Cref{thm:improve_comp} we can find a constant-depth \QACZ~circuit $C_i^*$ such that 
    \begin{align}
        \frobd \left( W_i, C_i^* \right) \leq \frac{1}{\poly(n)}.
    \end{align}
    Performing this for all $n$ measurement qubits requires quasi-polynomial computational complexity. By \Cref{thm:sew_order}, we can find a sewing order for all the $C_i^*$ circuits, in polynomial time, which sews them into the poly-logarithmic depth circuit $C_\text{sew}^*$. 
    
    To conclude, we will prove that  $C_\text{sew}^*$ is $1/\poly(n)$-close to $C \otimes C^\dagger$. By triangle inequality and \Cref{thm:main_result}
    \begin{align}
        \frobd(C_\text{sew}^*,C \otimes C^\dagger) & \leq \frobd(C_\text{sew}^*,C_\text{sew}) + \frobd(C_\text{sew},C \otimes C^\dagger) \leq \frobd(C_\text{sew}^*,C_\text{sew})+\frac{1}{\poly(n)} 
    \end{align}
    Therefore, all that remains is to show that  $C_\text{sew}^*$ is $1/\poly(n)$-close to $C_\text{sew}$. This can be achieved by leveraging the bound of \Cref{thm:improve_comp} and a simple hybrid argument, 
    \begin{align}
        \frobd(C_\text{sew}^*,C_\text{sew}) & = \frobd\left(\prod_{i} C^*_i,\prod_{i=1}^n W_i\right) \\
        &\leq \frobd\left(\prod_{i=1}^n C^*_i,W_1\prod_{i=2}^n C^*_i\right)+\frobd\left(W_1\prod_{i=2}^n C^*_i,\prod_{i=1}^n W_i\right)  \\
        &\leq \frobd\left(C^*_1,W_1\right)+\frobd\left(\prod_{i=2}^n C^*_i,\prod_{i=2}^n W_i\right)  \\
        &\leq \frac{1}{\poly(n)}+\frobd\left(\prod_{i=2}^n C^*_i,\prod_{i=2}^n W_i\right) \\
        & \quad \quad \quad \vdots \\
        &\leq n \cdot \frac{1}{\poly(n)} \leq \frac{1}{\poly(n)}.
    \end{align}
    This concludes the proof.
\end{proof}

\section{Concentration and Learning of \texorpdfstring{\QACZ}~ with Ancillas} \label{sec:ancillas}

Now we will show how things change for  \QACZ~circuits with ancillas. In this section, let $C$ be the unitary corresponding to an $(n+a)$-qubit \QACZ~circuit, of the same form as \Cref{eqn:qac0_circ}, operating on $n$ standard qubits and $a$ ancilla qubits. Similar to \cite{huang2024learning}, we will only consider circuits where the ancillas are initialized to the $\ket{0^a}$ state and the computation is clean (meaning ancillas are reverted to the $\ket{0^a}$ state at the end of the computation). Note that since the computation is clean, the action of $C$ on the $(n+a)$-qubit system is equivalent to the action of another unitary $A$ on just the $n$-qubit system without ancillas, i.e.
\begin{align} \label{eqn:clean_compute}
    C (I \otimes \ket{0^a}) = A \otimes \ket{0^a}.
\end{align}
We will define the Heisenberg-evolved Pauli observables of this system ``without ancilla restriction'' as
\begin{align}
    \obsanc = C (P_i \otimes I^a) C^\dagger 
\end{align}
and ``with ancilla restriction'' as
\begin{align}
    \obsared = (I \otimes \bra{0^a}) \cdot \obsanc \cdot (I \otimes \ket{0^a}) = (I \otimes \bra{0^a}) C (P_i \otimes I^a) C^\dagger (I \otimes \ket{0^a}) = A P_i A^\dagger.
\end{align}

\subsection{Concentration of \texorpdfstring{\QACZ}~ Heisenberg-Evolved Observables with Ancillas} \label{sec:anc_conc}
We will now re-prove the concentration results of \Cref{sec:heis_evolve_obs} for  \QACZ~circuits with ancillas.  To begin, generalizing \Cref{thm:avg_dist} to \QACZ~circuits with ancillas, we bound the distance between the true observable $\obsared$ and the observable $\obsared^*$ (with large $CZ$ gates removed).

\begin{lemma} \label{thm:avg_dist_anc}
    Let $C$ be an $(n+a)$-qubit \QACZ~circuit performing clean computation with respect to $a$ ancillas. Let $\obsanc = C (P_i \otimes I^a) C^\dagger$ be the Heisenberg-evolved $P_i$ observable without ancilla restriction and $\obsared = (I \otimes \bra{0^a}) \cdot \obsanc \cdot (I \otimes \ket{0^a})$ be the same observable with ancilla restriction. Define $\obsanc^*= \widetilde{C} (P_i \otimes I^a) \widetilde{C}^\dagger$ to be the Heisenberg-evolved $P_i$ observable corresponding to the \QACZ~circuit $\widetilde{C}$, which is simply $C$ with all $m$ $CZ_k$'s of size $k> \gaterem$ removed. Let $\obsared^* = (I \otimes \bra{0^a}) \cdot \obsanc^* \cdot (I \otimes \ket{0^a})$ be the ancilla-restricted version of this observable. The average-case distance between observables corresponding to circuits $C$ and $\widetilde{C}$ is upper-bounded by
    \begin{align}
        \frobd \left( \obsared, \obsared^*\right) \leq \frobd ( \obsanc, \obsanc^* ) \leq 2^a \cdot \frac{9m^2}{2^{\gaterem}}.
    \end{align}
\end{lemma}
\begin{proof}
    We derive the following bound, using the facts that $\|AB\|_F \leq \|A\| \cdot \|B\|_F$ and $\|I \otimes \ket{0^a}\| = 1$:
    \begin{align}
        \| \obsared - \obsared^* \|_F &= \| (I \otimes \bra{0^a}) \left( C (P_i \otimes I^a)  C^\dagger - \widetilde{C} (P_i \otimes I^a) \widetilde{C}^\dagger \right)(I \otimes \ket{0^a}) \|_F \\
        & \leq  \| I \otimes \bra{0^a} \| \cdot \|\left( C (P_i \otimes I^a)  C^\dagger - \widetilde{C} (P_i \otimes I^a) \widetilde{C}^\dagger \right)(I \otimes \ket{0^a}) \|_F \\
        & \leq  \| I \otimes \bra{0^a} \| \cdot \| C (P_i \otimes I^a)  C^\dagger - \widetilde{C} (P_i \otimes I^a) \widetilde{C}^\dagger \|_F \cdot \|I \otimes \ket{0^a} \| \\
        & =  1 \cdot \| C (P_i \otimes I^a)  C^\dagger - \widetilde{C} (P_i \otimes I^a) \widetilde{C}^\dagger \|_F \cdot 1 \\
        & = \| C (P_i \otimes I^a)  C^\dagger - \widetilde{C} (P_i \otimes I^a) \widetilde{C}^\dagger \|_F \\
        & = \| \obsanc - \obsanc^* \|_F. 
    \end{align}
    Therefore, if $C$ has $m$ $CZ$ gates of width $\geq \gaterem$ to be removed in $\widetilde{C}$, then following from \Cref{eqn:obs_ub},
    \begin{align}
        \frobd ( \obsared,  \obsared^* ) & \leq \frobd ( \obsanc, \obsanc^* ) \\
        & \leq \frac{1}{2^n} \left( \sum_{j=1}^m 3\cdot 2^{(n+a-k_j)/2} \right)^2 \\
        & \leq \frac{1}{2^n} \left( 3m \cdot \max_j  2^{(n+a-k_j)/2} \right)^2   \\
        & = \frac{9m^2}{2^{\gaterem-a}}.
    \end{align}
\end{proof}
We will now demonstrate that the Pauli weight of the observable with ancilla restriction $\obsared$ can be upper-bounded by that of the observable without ancilla restriction $\obsanc$. 
\begin{lemma} \label{thm:anc_weight_relate} Let $\calS \subseteq \pauli^n$ be a subset of the set of $n$-qubit Paulis, then 
    \begin{align}
        \wt^{\in \calS}[\obsared] \leq 2^a \cdot \wt^{\in\calS}[\obsanc].
    \end{align}
\end{lemma}
\begin{proof}[Proof of \Cref{thm:anc_weight_relate}]
    We will begin by assuming that $\obsanc$ has the following Pauli decomposition,
    \begin{align}
        \obsanc = \sum_{Q \in \pauli^n, R \in \pauli^a} \widehat{\obsanc}(Q \otimes R) \cdot Q \otimes R
    \end{align}
    We can relate the Pauli coefficients of $\obsanc$ to those of $\obsared$, as
    \begin{align}
        \widehat{\obsared}(S) &= \frac{1}{2^n} \cdot \Tr \left((I \otimes \bra{0^a}) \cdot \obsanc \cdot (I \otimes \ket{0^a}) \cdot S\right) \\
        &= \frac{1}{2^n} \cdot \Tr \left((I \otimes \bra{0^a}) \left(\sum_{Q \in \pauli^n, R \in \pauli^a} \widehat{\obsanc}(Q \otimes R) \cdot Q \otimes R \right) (I \otimes \ket{0^a}) \cdot (S \otimes I^a)\right) \\
        &= \frac{1}{2^n} \cdot \sum_{Q \in \pauli^n, R \in \pauli^a} \widehat{\obsanc}(Q \otimes R) \cdot \Tr \left(QS \otimes \bra{0^a}R \ket{0^a}\right) \\
        &= \frac{1}{2^n} \cdot \sum_{Q \in \pauli^n, R \in \pauli^a} \widehat{\obsanc}(Q \otimes R) \cdot \Tr (QS) \cdot \bra{0^a}R \ket{0^a} \\
        &= \frac{1}{2^n} \cdot \sum_{Q \in \pauli^n, R \in \pauli^a} \widehat{\obsanc}(Q \otimes R) \cdot 2^n \cdot \delta\{Q=S\} \cdot \delta\{R \in \{I,Z\}^{\otimes a}\} \\
        &= \sum_{R \in \{I,Z\}^{\otimes a}} \widehat{\obsanc}(S \otimes R).
    \end{align}
    Leveraging this result and the Cauchy-Schwarz inequality, we can prove the desired result as
    \begin{align}
        \wt^{\in\calS}[\obsared] &= \sum_{Q:Q\in \calS} \left|\widehat{\obsared}(Q)\right|^2 \\
        &=\sum_{Q:Q\in \calS} \left|\sum_{R \in \{I,Z\}^{\otimes a}} \widehat{\obsanc}(Q \otimes R)\right|^2 \\
        & \leq 2^a \cdot \sum_{Q:Q\in \calS} \sum_{R \in \{I,Z\}^{\otimes a}} \left|\widehat{\obsanc}(Q \otimes R)\right|^2 \\
        & \leq 2^a \cdot \wt^{\in \calS}[\obsanc].
    \end{align}
\end{proof}
\noindent Note that this result is similar to that of \cite[Proposition 27]{nadimpalli2023pauli}, but we prove the bound for the weight over arbitrary subsets of Paulis $\calS \in \pauli^n$, as opposed to just the weight above a certain degree. This enables us to use the result to achieve both the low-degree and low-support concentration bounds as simple corollaries.

\begin{corollary}[Low-Degree Concentration with Ancillas] \label{thm:deg_conc_anc}
    Suppose $C$ is a depth-$d$, size-$s$ \QACZ~circuit performing clean computation on $n+a$ qubits, where $a$ is the number of ancilla. Let $\obsared = (I \otimes \bra{0^a}) C (P_i \otimes I^a) C^\dagger (I \otimes \ket{0^a})$ be a Heisenberg-evolved single-qubit Pauli observable with ancilla restriction. Then for every degree $k \in [n]$,
    \begin{align}
        \wt^{>k}[\obsared] \leq \calO \left(s^2 2^{-k^{1/d}}\right) \cdot 2^a.
    \end{align}
\end{corollary}
\begin{proof}[Proof of \Cref{thm:deg_conc_anc}]
    Let $\calS$ be the set of Paulis of degree $>k$, i.e. $\calS=\{P \in \paulis:|P|>k\}$. By \Cref{thm:anc_weight_relate},
    \begin{align} \label{eqn:weight_relate}
        \wt^{>k}[\obsared] \leq 2^a \cdot \wt^{>k}[\obsanc].
    \end{align}
    Since $\obsanc$ is an observable defined over $n+a$ qubits without ancilla restriction, by \Cref{thm:low_deg_conc},
    \begin{align}
        \wt^{>k}[\obsanc] \leq \calO \left(s^2 2^{-k^{1/d}}\right).
    \end{align}
    Plugging this into \Cref{eqn:weight_relate}, we obtain the desired result.
\end{proof}

\begin{corollary}[Low-Support Concentration with Ancillas] \label{thm:sup_conc_anc}
    For $\calS^* = \supp(\obsared^*)$, the weight of $\obsared$ outside the support of $\obsared^*$ is upper-bounded as
    \begin{align}
        \wt^{\notin\calS^*}[\obsared] \leq \frobd \left( \obsared, \obsared^*\right) \leq \frobd \left( \obsanc, \obsanc^*\right) \leq 2^a \cdot \frac{9m^2}{2^{\gaterem}} 
    \end{align}
\end{corollary}
\begin{proof}[Proof of \Cref{thm:sup_conc_anc}]
    Note that the weight expression satisfies the triangle inequality. Thus, the weight of $\obsared$ can be decomposed with respect to  $\obsared^*$ as follows. 
    \begin{align}
        \wt^{\notin\calS^*}[\obsared] &= \wt^{\notin\calS^*}\left[\obsared^*+ (\obsared-\obsared^*)\right] \\
        &\leq \left( \wt^{\notin\calS^*}\left[\obsared^*\right]^{1/2}+ \wt^{\notin\calS^*}\left[\obsared-\obsared^*\right]^{1/2}\right)^2 \label{eqn:inter_weight}
    \end{align}
    Since $\calS^* = \supp(\obsared^*)$, by \Cref{thm:weight_zero_support},
    \begin{align}
        \wt^{\notin\calS^*}\left[\obsared^*\right] =0.
    \end{align}
    Furthermore, by \Cref{thm:avg_dist_anc}
    \begin{align}
         \wt^{\notin\calS^*}\left[\obsared-\obsared^*\right]\leq \frobd \left( \obsared, \obsared^*\right) \leq \frobd \left( \obsanc, \obsanc^*\right).
    \end{align}
    Plugging these results into \Cref{eqn:inter_weight} and leveraging our upper-bound on the distance from \Cref{thm:avg_dist_anc} obtains the desired expression.
\end{proof}

\subsection{Learning \texorpdfstring{\QACZ}~ Heisenberg-Evolved Observables with Ancillas}
In order to learn the unitary of an $(n+a)$-qubit \QACZ~circuit with ancillas, it turns out that we can use the procedure of \Cref{sec:learn} with some slight modifications. 

\begin{proposition}[Learning \QACZ~with Logarithmic Ancillas] \label{thm:learn_log_anc}
    Suppose we are given an $(n+a)$-qubit depth-$d$ \QACZ~circuit governed by unitary $C$, performing clean computation 
    \begin{align} \label{eqn:clean_comp}
        C (I \otimes \ket{0^a}) = A \otimes \ket{0^a}
    \end{align}
    on a logarithmic number of ancilla qubits, i.e. $a=\calO( \log n)$.
    For error parameter $\varepsilon= 1/\poly(n)$ and failure probability $\delta \in (0,1)$, we can learn a $2n$-qubit unitary $A_\text{sew}$ which is $\varepsilon$-close to the unitary $A \otimes A^\dagger$, i.e.
    \begin{align}
        \davg(A_\text{sew}, A \otimes A^\dagger) \leq \varepsilon,
    \end{align}
    with high probability $1-\delta$, as well as quasi-polynomial sample and computational complexity.
\end{proposition}
\begin{proof}[Proof of \Cref{thm:learn_log_anc}]
    As expressed in \Cref{eqn:clean_comp}, since the circuit $C$ performs clean computation, the ancillas return to their starting state by end of the computation. Therefore, we can measure out the ancillas at the end of the computation, to obtain the unitary $A$ which performs the same computation as $C$, but only acts on the $n$ computation qubits. Thus, in the learning procedure, we will be interested in learning and sewing the $n$-qubit Heisenberg-evolved single-qubit Paulis \emph{with} ancilla restriction, i.e. $\obsared$.

    Since these observables with ancilla restriction only act on $n$ qubits, we can straightforwardly apply \Cref{alg:learning_alg} for learning the Heisenberg-evolved single-qubit Pauli observables of the circuit. However, the learning guarantees in this case differ from those of the ancilla-free case. In particular, to show that the learned observable $\widetilde{O}^{(\support)}_{P_i,n}$ is close to the desired observable with ancilla restriction $\obsared$, we need to bound the distance between these operators. Plugging $\widetilde{O}^{(\support)}_{P_i,n}$ and  $\obsared$  into \Cref{thm:learn_dist}, we get that their distance is bounded as
    \begin{align}
        \frobd \left( \widetilde{O}^{(\support)}_{P_i,n},  \obsared \right) \leq 2 \cdot  4^\support \cdot \eta^2 + \wt^{\notin\calS^*}[\obsared] \leq 2 \cdot 4^\support \cdot \eta^2 + 2^a \cdot \frac{9m^2}{2^{\gaterem}},
    \end{align}
    where $\eta$ is the learning accuracy and the second inequality leveraged our low-support concentration result from \Cref{thm:sup_conc_anc}. Notice that the key difference to the ancilla-free case is a $2^a$ factor amplifying the error in removing large $CZ$ gates.
    
    As in the ancilla-free case, we need to perform a balancing act to ensure that this learning error is small while ensuring that the algorithm has efficient sample and computational complexity. Since the sample and time complexity are directly related to the supports of the learned observables, to achieve quasi-polynomial complexity, we will want to learn observables with $\poly\log n$ support. However, if we were to remove all $CZ$ gates of width $\geq \gaterem = c \cdot \log(n+a)$ in this circuit of size $s=\poly(n)$, then 
    \begin{align} \label{eqn:error_ub}
        \wt^{\notin\calS^*}[\obsared] & \leq \frac{9m^2}{2^{\gaterem-a}}  \leq \frac{9\cdot s^2\cdot 2^a}{2^{\log((n+a)^c)}} = \frac{9\cdot s^2\cdot 2^a}{(n+a)^c}.
    \end{align}
    Due to the $2^a$ factor in the numerator of \Cref{eqn:error_ub}, $\wt^{\notin \calS^*}[O_{P_i}]$ is guaranteed to be $\leq1/\poly(n)$ only if the number of ancillas is logarithmic, i.e. $a = \calO(\log n)$. 
\end{proof}

Note that the reason our algorithm is restricted to \QACZ~circuits with a logarithmic number of ancillas is the ancilla-dependence in the low-support concentration result of \Cref{thm:sup_conc_anc}. However, we do not believe this bound to be tight. If \cite[Conjecture 1]{nadimpalli2023pauli} were proven true, it would imply ancilla-independent low-degree concentration of the Heisenberg-evolved observables.
\begin{conjecture}[\cite{nadimpalli2023pauli} Conjecture 1] \label{conj:low_deg_choi_conc} For a size-$s$, depth-$d$ \QACZ~circuit acting on $n$-qubits and $\poly(n)$ ancilla qubits, then for all $k \in [n+1]$,
    \begin{align}
        \wt^{>k} [\Phi_{\calE_C}]  \leq \poly(s) \cdot 2^{-\Omega(k^{1/d})}.
    \end{align} 
\end{conjecture}
\begin{corollary}
    For a size-$s$, depth-$d$ \QACZ~circuit acting on $n$-qubits and $\poly(n)$ ancilla qubits, then for all degrees $k \in [n]$,
    \begin{align}
        \wt^{>k} [\obsared] \leq \poly(s) \cdot 2^{-\Omega(k^{1/d})}.
    \end{align}
\end{corollary}
\begin{proof}
    By  \Cref{thm:relate_coeffs} and \Cref{conj:low_deg_choi_conc}, for $k \in [n]$
    \begin{align}
        \wt^{>k} [\obsared] &= \sum_{Q \in \pauli^{n}:|Q| >k} \left|\widehat{\obsared}(Q)\right|^2  \\
        &= 4  \sum_{Q \in \pauli^{n}:|Q| >k} \left| \widehat{\Phi_{\calE_C}}(Q \otimes P)\right|^2 \\
        &\leq 4  \sum_{R \in \pauli^{n+1}:|R| > k+1} \left| \widehat{\Phi_{\calE_C}}(R)\right|^2 \\
        & \leq 4  \cdot \poly(s) \cdot 2^{-\Omega((k+1)^{1/d})} \\ 
        &\leq \poly(s) \cdot 2^{-\Omega(k^{1/d})}
    \end{align}
\end{proof}
\noindent However, for the purposes of our learning algorithm, it does not suffice to have low-degree concentration. Instead, we need low-support concentration. Thus, we conjecture that the ancilla-dependence of the low-support concentration result of \Cref{thm:sup_conc_anc} can be eliminated.
\begin{conjecture}[Ancilla-Independent Low-Support Concentration] For a size-$s$, depth-$d$ \QACZ~circuit acting on $n$-qubits and $\poly(n)$ ancilla qubits and support $\calS$ such that $|\calS|=k^d$,
    \begin{align}
        \wt^{\notin\calS}[\obsared] \leq \poly(s) \cdot 2^{-\Omega(k^{1/d})}.
    \end{align}
    
\end{conjecture}
\noindent If this conjecture were proven true, it would imply that our learning algorithm works for \QACZ~circuits with polynomially many ancilla qubits.
\begin{corollary}[Efficient learning of \QACZ~with Polynomial Ancillas]
    Suppose we are given an $(n+a)$-qubit depth-$d$ \QACZ~circuit governed by unitary $C$, performing clean computation 
    \begin{align}
        C (I \otimes \ket{0^a}) = A \otimes \ket{0^a}
    \end{align}
    on polynomially many ancilla qubits, i.e. $a=\poly(n)$.
    For failure probability $\delta \in (0,1)$, we can learn a $2n$-qubit unitary $A_\text{sew}$ such that
    \begin{align}
        \davg(A_\text{sew}, A \otimes A^\dagger) \leq \frac{1}{\poly(n)},
    \end{align}
    with high probability $1-\delta$. The sample and computational complexity of this procedure are quasi-polynomial.
\end{corollary}
\begin{proof} Following from the proof of \Cref{thm:learn_log_anc}, in this case, the distance between  $\widetilde{O}^{(\support)}_{P_i,n}$ and  $\obsared$ is bounded as
    \begin{align}
        \frobd \left( \widetilde{O}^{(\support)}_{P_i,n}, \obsared \right) \leq 2 \cdot  4^\support \cdot \eta^2 + \wt^{\notin\calS^*}[\obsared] \leq 2 \cdot 4^\support \cdot \eta^2 +  \poly(s) \cdot 2^{-\Omega(k^{1/d})},
    \end{align}
    where $\eta$ is the learning accuracy. Since the sample and time complexity are directly related to the supports of the learned observables, to achieve quasi-polynomial complexity, we will want to learn observables with $\poly\log n$ support. In this case, if we remove all $CZ$ gates of width $\geq \gaterem = c \cdot \log(n+a)$ where $a=n^b$ in this circuit of size $s=\poly(n)$, then 
    \begin{align}
        \wt^{\notin\calS^*}[\obsared] & \leq \poly(s) \cdot 2^{-\Omega(k^{1/d})} \leq \frac{\poly(s)}{2^{c \cdot \log(n+n^b)}}\leq \frac{\poly(s)}{n^{bc}}.
    \end{align}
    Thus, if $c$ is chosen to be a constant sufficiently large such that $n^{bc} > \poly(s)$, then $\wt^{\notin\calS^*}[\obsared] \leq 1/\poly(n)$ and we obtain the desired learning guarantee.
\end{proof}

\section{Hardness of Learning \texorpdfstring{\QACZ}~} \label{sec:hardness_proof}

We will conclude with a result on the hardness of learning~\QACZ. This result follows straightforwardly from \cite[Proposition 3]{huang2024learning}, which showed that it is exponentially-hard to learn \QACZ~according to the diamond-norm distance. In particular, the proof leverages a specific worst-case logarithmic-depth circuit $U_x$, which via a Grover lower-bound is shown to require exponential queries to learn according to the worst-case measure. We simply observe that $U_x$ is in and of itself a \QACZ~circuit to extend the hardness result to \QACZ~circuits. For the sake of completeness, we re-write the theorem statement of \cite[Proposition 3]{huang2024learning} in the context of \QACZ~circuits, as well as the proof.

\begin{proposition}[Hardness of learning \QACZ] Consider an unknown $n$-qubit unitary $U$ generated by a \QACZ~circuit. Then,
\begin{enumerate}
    \item Learning $U$ to $\frac{1}{3}$ dimond distance with high probability requires $\exp(\Omega(n))$ queries.
    \item Distinguishing whether $U$ equals the identity matrix $I$ or is $\frac{1}{3}$-far from the identity matrix in diamond distance with high probability requires $\exp(\Omega(n))$ queries.
\end{enumerate}
\end{proposition}
\begin{proof}
    For $x,y \in \{0,1\}^n$, let $U_x$ be the unitary,
    \begin{align}
        U_x \ket{y} = \begin{cases}
            1, & x=y, \\
            -1 & x \neq y,
        \end{cases}
    \end{align}
    which can be constructed as 
    \begin{align}
        U_x = \left(\prod_{\substack{i \in [n]: \\ x_i =0}} X_i\right) CZ_{[n]} \left(\prod_{\substack{i \in [n]: \\ x_i =0}} X_i\right).
    \end{align}
    
    \cite{huang2024learning} used $U_x$ to prove a learning lower-bound for logarithmic depth circuits comprised of constant-width gates. To begin, they showed that, in the class of circuits comprised solely of constant-width gates, $U_x$ could be synthesized in $\calO(\log n)$-depth. Core to the argument is that, if one can learn a unitary $U$ up to $\frac{1}{3}$ error in diamond distance with high probability or distinguish whether $U$ equals the identity $I$ or is $\frac{1}{3}$-far from $I$ in diamond distance with high probability, then one can successfully distinguish $I$ from $U_x$. However, distinguishing $I$ from one of $U_x$, $\forall x \in \{0,1\}$ is the Grover search problem. Therefore, by the Grover lower bound \citep{bennett1997strengths}, the number of queries must be at least $\Omega(2^{n/2})=\exp(\Omega (n))$.

    Note that the unitary $U_x$ is not only contained in the class of logarithmic-depth constant-width circuits, but also is contained in \QACZ, since it is comprised solely of single-qubit gates and an $n$-qubit $CZ$ gate. Therefore, the argument by and lower-bound of \cite{huang2024learning} also applies to \QACZ, concluding the proof.
\end{proof}

\end{document}